\documentclass[aps,pra,amsmath,amssymb,superscriptaddress,twocolumn,notitlepage,nofootinbib,longbibliography]{revtex4-2}

\usepackage[final]{graphicx}
\usepackage{times,bbm,amsmath,amssymb,amsthm}
\usepackage{epsfig,color}
\usepackage{xcolor}
\usepackage[colorlinks=true,citecolor=blue,linkcolor=blue]{hyperref}
\hypersetup{
    allcolors  = {blue},
}

\usepackage{cleveref}
\usepackage{float,siunitx}
\usepackage[caption = false]{subfig}
\usepackage{verbatim}
\usepackage[greek,english]{babel}
\usepackage{thumbpdf,enumerate}
\usepackage{booktabs}
\usepackage{sidecap}
\usepackage[scaled=.8]{couriers}    
\usepackage{pstricks}
\usepackage{multirow}
\usepackage{placeins}
\usepackage{relsize}  
\usepackage{pst-grad,bm}
\usepackage{epigraph}
\usepackage{gensymb}
\usepackage{longtable}
\usepackage{booktabs}
\usepackage{gensymb}

\usepackage{soul}
\usepackage{ulem} 
\usepackage{acronym}
\usepackage{physics}
\usepackage{tikz}

\theoremstyle{plain}

\newtheorem*{theorem*}    {Theorem}
\newtheorem*{proposition*}{Proposition}
\newtheorem*{lemma*}      {Lemma}
\newtheorem*{corollary*}  {Corollary}
\newtheorem*{conjecture*} {Conjecture}

\begin{document}
\title{Experimental certification of contextuality, coherence and dimension in a programmable {universal photonic processor}}




\author{Taira Giordani} 
\thanks{These two authors contributed equally}
\affiliation{Dipartimento di Fisica, Sapienza Universit\`{a} di Roma, Piazzale Aldo Moro 5, I-00185 Roma, Italy}

\author{Rafael Wagner}
\thanks{These two authors contributed equally}
\affiliation{International Iberian Nanotechnology Laboratory (INL)
 Av. Mestre José Veiga s/n, 4715-330 Braga, Portugal
}
\affiliation{Centro de F\'{i}sica, Universidade do Minho, Campus de Gualtar, 4710-057 Braga, Portugal}

\author{Chiara Esposito} 
\affiliation{Dipartimento di Fisica, Sapienza Universit\`{a} di Roma, Piazzale Aldo Moro 5, I-00185 Roma, Italy} 

\author{Anita Camillini}
\affiliation{International Iberian Nanotechnology Laboratory (INL)
 Av. Mestre José Veiga s/n, 4715-330 Braga, Portugal
}
\affiliation{Centro de F\'{i}sica, Universidade do Minho, Campus de Gualtar, 4710-057 Braga, Portugal}

\author{Francesco Hoch} 
\affiliation{Dipartimento di Fisica, Sapienza Universit\`{a} di Roma, Piazzale Aldo Moro 5, I-00185 Roma, Italy} 

\author{Gonzalo Carvacho} \affiliation{Dipartimento di Fisica, Sapienza Universit\`{a} di Roma, Piazzale Aldo Moro 5, I-00185 Roma, Italy}

\author{Ciro Pentangelo}
\affiliation{Dipartimento di Fisica, Politecnico di Milano, Piazza Leonardo da Vinci, 32, I-20133 Milano, Italy}
\affiliation{Istituto di Fotonica e Nanotecnologie, Consiglio Nazionale delle Ricerche (IFN-CNR), 
Piazza Leonardo da Vinci, 32, I-20133 Milano, Italy}

\author{Francesco Ceccarelli}
\affiliation{Istituto di Fotonica e Nanotecnologie, Consiglio Nazionale delle Ricerche (IFN-CNR), 
Piazza Leonardo da Vinci, 32, I-20133 Milano, Italy}

\author{Simone Piacentini}
\affiliation{Istituto di Fotonica e Nanotecnologie, Consiglio Nazionale delle Ricerche (IFN-CNR), 
Piazza Leonardo da Vinci, 32, I-20133 Milano, Italy}

\author{Andrea Crespi}
\affiliation{Dipartimento di Fisica, Politecnico di Milano, Piazza Leonardo da Vinci, 32, I-20133 Milano, Italy}
\affiliation{Istituto di Fotonica e Nanotecnologie, Consiglio Nazionale delle Ricerche (IFN-CNR), 
Piazza Leonardo da Vinci, 32, I-20133 Milano, Italy}

\author{Nicol\`o Spagnolo} \affiliation{Dipartimento di Fisica, Sapienza Universit\`{a} di Roma, Piazzale Aldo Moro 5, I-00185 Roma, Italy}


\author{Roberto Osellame}
\affiliation{Istituto di Fotonica e Nanotecnologie, Consiglio Nazionale delle Ricerche (IFN-CNR), 
Piazza Leonardo da Vinci, 32, I-20133 Milano, Italy}


\author{Ernesto F. Galv\~ao}
\email[Corresponding author: ]{ernesto.galvao@inl.int}
\affiliation{International Iberian Nanotechnology Laboratory (INL)
 Av. Mestre José Veiga s/n, 4715-330 Braga, Portugal
}

\affiliation{Instituto de F\'isica, Universidade Federal Fluminense, Av. Gal. Milton Tavares de Souza s/n, Niter\'oi, RJ, 24210-340, Brazil}

\author{Fabio Sciarrino}
\email[Corresponding author: ]{fabio.sciarrino@uniroma1.it}
\affiliation{Dipartimento di Fisica, Sapienza Universit\`{a} di Roma, Piazzale Aldo Moro 5, I-00185 Roma, Italy}

\begin{abstract}
Quantum superposition of high-dimensional states {enables} both computational speed-up and security in cryptographic protocols. However, the exponential complexity of tomographic processes makes certification of these properties a challenging task. 
In this work, we experimentally certify coherence witnesses tailored for quantum systems of  increasing dimension, using pairwise overlap measurements enabled by a {six-mode universal photonic processor fabricated with a femtosecond laser writing technology}. In particular, we show the effectiveness of the proposed coherence and dimension witnesses for qudits of dimensions up to 5. We also demonstrate  advantage in a quantum interrogation task, and show it is fueled by quantum contextuality. Our experimental results testify to the efficiency of this novel approach {for} the certification of quantum properties in {programmable integrated photonic platforms}. 

\end{abstract}

\maketitle

\section*{Introduction}

Quantum computers are capable of solving problems believed to be effectively impossible classically, such as sampling from complex probability distributions~\cite{aaronson2011computational}, predicting properties of physical systems~\cite{huang2020predictingmany}, and factoring large integers~\cite{shor1999polynomial}. 
The quantum computational advantage for these tasks is built on rigorous no-go results in computational complexity theory, showing a gap between quantum and classical resources for the same task that can be exponential. 
Research on quantum foundations also addresses quantum advantage, by asking the question: what type of quantum information processing \textit{cannot} be explained classically? Answers  broadly suggest a different kind of advantage, not in terms of computational power, but in terms of intrinsic classical limits to information processing. 

Advantage in quantum information processing pushes success rates in communication tasks~\cite{saha2019preparation,spekkens2009parity, Buhrman10}, security of key distribution protocols~\cite{ekert1992quantumcryptography}, and success rates in discrimination tasks~\cite{schmid2018contextual} beyond those that can be reached using only classical resources. This kind of advantage results from understanding quantum foundational aspects of nonclassical resources such as entanglement~\cite{horodecki2009quantumentanglement}, coherence~\cite{streltsov2017quantumcoherence},  Bell nonlocality~\cite{brunner2014nonlocality} and contextuality~\cite{budroni2021kochen}, where classical-quantum gaps in explaining the phenomena can be predicted, bounding success rates in a quantifiable manner. For instance, phenomena that can be reproduced by  noncontextual models~\cite{spekkens2005contextuality} include interference~\cite{catani2021whyinterference}, superdense coding~\cite{spekkens2007toymodel} and Gaussian quantum mechanics~\cite{bartlett2012reconstruction}, but it is possible to describe precisely when processing of quantum information allows an advantage over such seemingly powerful models in many tasks~\cite{spekkens2009parity,saha2019preparation}. 

In this work, we propose and test families of coherence and contextuality witnesses in { proof-of-principle device dependent }experiments carried out with single-photon states processed by a {programmable} integrated photonic circuit. Such witnesses require the careful preparation of the system in a set of different states and then the estimation of a set of pairwise state overlaps \cite{galvaobrod2020quantum, giordani2021witnessing, wagner2022coherence, Wagner2022}. For this task, we encode qubits and qudits with dimensions up to five in {a six-mode universal photonic processor (UPP) realized with the femtosecond laser writing technology \cite{corrielli2021flm}.} 

We start by testing a recent 
quantum information advantage~\cite{wagner2022coherence} for the task of quantum interrogation, first proposed as the celebrated Elitzur-Vaidman bomb-testing experiment~\cite{elitzur1993quantum}. This task has had a profound impact in quantum foundations~\cite{vaidman2003themeaning,hardy1992emptywaves} that later converged into technical developments, such as the possibility of performing counterfactual quantum computation~\cite{hosten2006counterfactual,liu2012experimental}, the development of the field of quantum imaging with undetected photons~\cite{lemos2014quantumimaging,lahiri2015theory}, and high-efficiency interrogation using the quantum Zeno effect~\cite{kwiat1995interactionfree,rudolph2000better}. We experimentally verify that the efficiency achievable by quantum theory cannot be explained by noncontextual models such as those of Refs.~\cite{spekkens2007toymodel,catani2021whyinterference}, as predicted in Ref.~\cite{wagner2022coherence}. We hence certify both the presence of this nonclassical resource in the device, and its ability to use it for information processing advantage.

Although noncontextual models are capable of capturing some aspects of quantum coherence~\cite{spekkens2007toymodel}, in a way similar to how local models reproduce some aspects of quantum entanglement~\cite{werner1989quantum}, coherence is still of utmost relevance for quantum information science. It  plays a major role in Shor's factoring algorithm~\cite{ahnefeld2022coherence}, and is crucial for the quantum advantage provided by linear-optical devices. One can then ask the question of what \textit{cannot} be explained with coherence-free models. The recently established inequalities of Refs.~\cite{galvaobrod2020quantum,Wagner2022} provide a precise answer, by rigorously bounding such models.

We show theoretically, numerically, and experimentally that a family of inequalities introduced in Ref.~\cite{wagner2022coherence} has a particularly interesting property: violations of such inequalities witness not only coherence inside the interferometers, but also the \textit{dimensionality} of the information encoded. Since Hilbert space dimension is itself a resource, 
the question of what can be done \textit{only} with qu\textit{dits} is of  relevance for information processing. Some information tasks do require qudits~\cite{vidick2011doesignorance,kewming2020hidingignorance,tavakoli2015quantum,tavakoli2015secret}, or have their security linked to the dimension \cite{Acin2006fromBelltoSecure}, so an important research field is devoted to developing methods to guarantee lower bounds on the Hilbert space dimension attained by different physical systems~\cite{brunner2008testing,wehner2008stephanie,gallego2010device,guhne2014bounding,ray2021graph}. 

The fact that this family of inequalities has violations only for \textit{coherent qudits} marks a new paradigm for quantum coherence allowed by the basis-independent perspective~\cite{designolle2021setcoherence,galvaobrod2020quantum}. There exists quantum coherence that is achieved by qudits that cannot be achieved with qubits. Such a fact has no precedent from the resource theoretic perspective of basis-dependent coherence~\cite{streltsov2017quantumcoherence}. Coherence captured only with qudits was first considered in Ref.~\cite{brunner2013dimension}.  We experimentally witness this proposed new form of coherence for qutrits, ququarts and ququints inside a six-mode programmable integrated universal photonic processor. We prove that the inequality violated by pure qutrits cannot be violated by qubits, complementing this result with numerical and experimental investigations, and we perform a similar analysis for our inequalities violated by ququart and ququint systems. In doing so, we extend the dimension witness result from Ref.~\cite{giordani2021witnessing} both qualitatively and quantitatively, making the best use of the flexibility and accuracy of our multi-mode {processor}.

\section*{Results}
\subsection{Theoretical framework}

Quantum coherence is commonly described as a basis-dependent property. Given some space $\mathcal{H}$ describing a  system and a fixed  basis $\Omega = \{\vert \omega \rangle \}_{\omega}$, any state $\rho$ is said to be \emph{coherent} if it is not diagonal with respect to $\Omega$. It is possible to avoid basis-dependence by considering sets of states~\cite{designolle2021setcoherence}. Given any set of states $\underline{\rho} = \{\rho_i\}_{i=0}^{n-1}$, the entire set is said to be basis-independent coherent, or simply set-coherent, if there exists no unitary $U$ such that $\underline{\rho} \mapsto \underline{\sigma} = \underline{U\rho U^\dagger} = \{U \rho_i U^\dagger \}_{i=0}^{n-1}$, with every $\sigma_i=U\rho_i U^\dagger$ diagonal. 

\begin{figure}[t]
    \centering
    \includegraphics[width=\columnwidth]{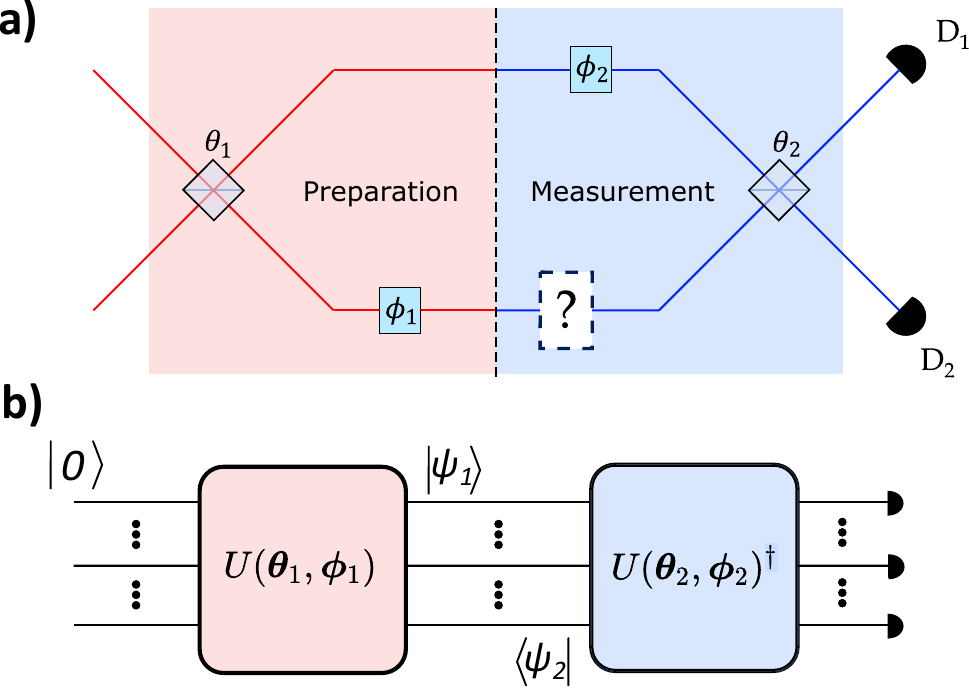}
    \caption{\textbf{Mach-Zehnder interferometer (MZI) and its multimode generalization for quantum interrogation, coherence, and dimension witnesses.} a) Any MZI can be ideally separated in a preparation stage (red), in which we prepare a qubit state $\ket{\psi(\theta_1,\phi_1)}$, and a measurement stage (blue) that projects onto another qubit state $\ket{\psi(\theta_2,\phi_2)}$. In a quantum interrogation experiment, the ?-box is an object that absorbs photons. b)  In analogy with the MZI, a multi-path interferometer encodes $d$-level systems by a set of {beam splitters $\boldsymbol{\theta_1}$} and {phase shifters $\boldsymbol{\phi_1}$} that operate on the $d$ mode as a unitary operator $U(\boldsymbol{\theta_1}, \boldsymbol{\phi_1})$. In the second stage, another round of {beam splitters $\boldsymbol{\theta_2}$} and {phase shifters $\boldsymbol{\phi_2}$} followed by a series of photodetectors detects photons at the $d$ output ports. With a single-photon input at the top input mode $\ket{0}$, this setup is capable of measuring two-state overlaps  $\vert \langle \psi_2 \vert \psi_1\rangle \vert^2=\vert \langle 0 \vert U(\boldsymbol{\theta}_2,\boldsymbol{\phi}_2)^\dagger  U(\boldsymbol{\theta}_1,\boldsymbol{\phi}_1) \vert 0\rangle \vert^2$. }
    \label{fig:MZI}
    \label{fig:high-dimension interference}
\end{figure}

\begin{figure*}[t]
    \centering
    \includegraphics[width=\textwidth]{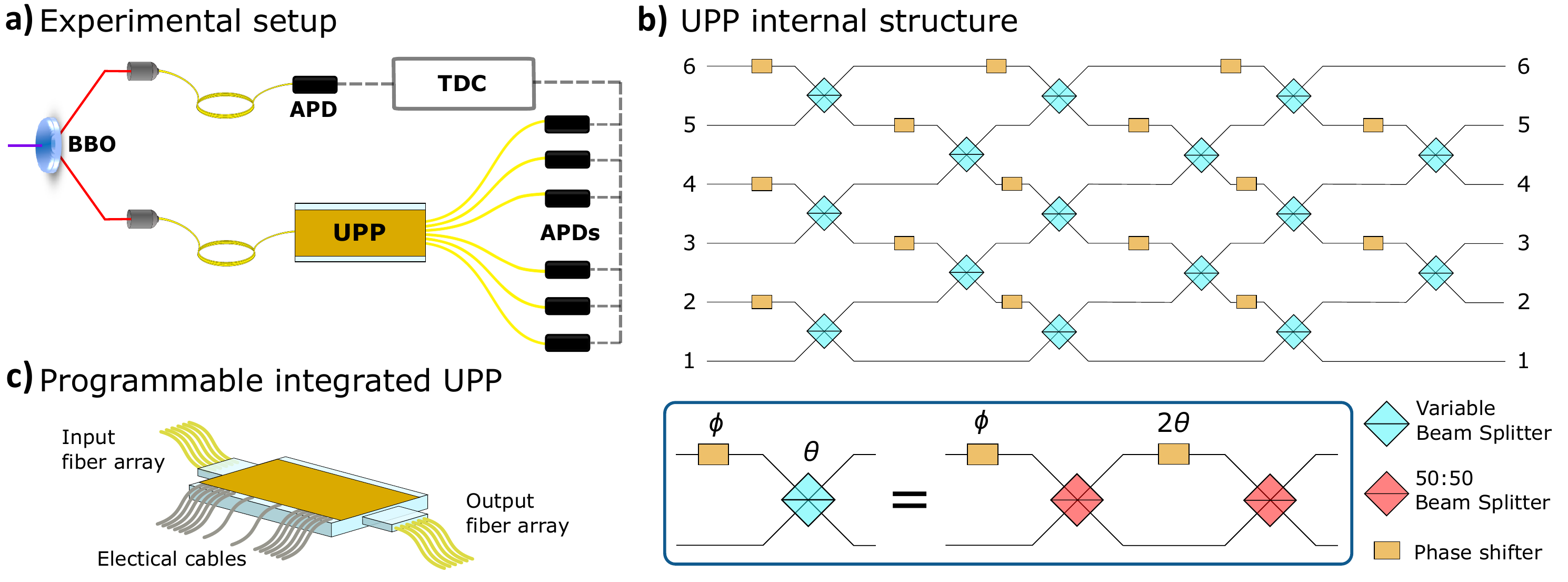}
    \caption{\textbf{Experimental setup and {universal photonic processor.}} {a) Experimental setup.} A pair of photons is generated by a SPDC source. One photon is detected as trigger. The second photon is sent in {a programmable integrated UPP} that can realize a generic unitary {transformation}. 
    At the output of the chip we measure the two-fold coincidence between the photon in the chip and the trigger photon. They are detected by exploiting {APDs}.  {b) UPP internal structure.} The optical circuit is a six-mode rectangular mesh of variable {beam splitters} and phase shifters{, enabling the implementation of arbitrary 6x6 unitary transformations.} Each variable {beam splitter} is actually a MZI structure with two 50:50 {beam splitters} and a phase shifter in between (see the inset). {c) Programmable integrated UPP}. The integrated device employed in the experiment is a 
    {UPP}, realized by the {femtosecond} laser writing technique in {an alumino-borosilicate} glass substrate. Two fiber arrays are directly plugged in at the input  and at the output of the interferometer. 
    {Thermo-optic phase shifters are patterned with the same technology on a thin gold layer deposited on the substrate. Electrical currents are supplied to the phase shifters through interposing printed circuit boards (not shown in figure for the sake of simplicity) allowing one to locally heat the waveguides and change the settings of the optical device.} 
    \textit{Legend}: BBO: Beta-Barium Borate crystal; APD: Avalanche photo-diode detector; TDC Time-to-digital converter.}
    \label{fig:chip}
\end{figure*}

Witnesses of such a notion of basis-independent coherence were proposed in Ref.~\cite{galvaobrod2020quantum}, building on the realization that set-coherence is a relational property among the states in $\underline{\rho}$. Bargmann invariants~\cite{bargmann1964note,oszmaniec2021measuring} completely characterize all the relational information of any set of states. The simplest such invariants are the two-state overlaps $r_{i,j}=\text{Tr}(\rho_i\rho_j)$, for $\rho_i,\rho_j \in \underline{\rho}$. 
In the Methods section and Supplementary Note 1 we recall why the overlap inequalities of Refs.~\cite{galvaobrod2020quantum,Wagner2022} serve as set-coherence witnesses. The first non-trivial inequality bounding coherence~\cite{galvaobrod2020quantum} was experimentally investigated in Ref.~\cite{giordani2021witnessing}, and bounds the three overlaps of a set of 3 states: 
\begin{equation}
    r_{0,1} + r_{0,2} - r_{1,2} \leq 1.\label{ineq: 3-cycle inequalities}
\end{equation}
Violations of such inequalities represent witnesses of basis-independent coherence of $\underline{\rho} = \{\rho_i\}_{i=0}^2$. However, as was shown in Ref.~\cite{Wagner2022}, this is \textit{also} a witness of  contextuality when we interpret each state either as an operational preparation procedure, or as a measurement effect. In Supplementary Note 2 we review in detail how these inequalities help to characterize contextuality.

As part of our certification we perform a quantum information task known as (standard) quantum interrogation~\cite{elitzur1993quantum}, that can be performed using a two-mode Mach-Zehnder interferometer (MZI) set-up, as depicted in Fig.~\ref{fig:MZI}a, and interpreted in light of our discussion about the connection between coherence and contextuality. For the purpose of testing our device, we will quantify the success rate of the interrogation task using the efficiency $\eta$ given by
\begin{equation}
    \eta = \frac{p_{succ}}{p_{succ}+p_{abs}},
    \label{eq:eff_bomb}
\end{equation}
where $p_{succ}$ is the probability of {successfully} detecting the presence of the object without it absorbing the photon, and $p_{abs}$ corresponds to the probability of absorption. In the Methods section we precisely describe the task, and how these probabilities relate to the MZI beam-splitting ratio. 
Ref.~\cite{wagner2022coherence} showed that noncontextual models \textit{cannot} explain $\eta$ for arbitrary beam-splitting ratios, and that there exists a quantifiable gap between the efficiencies achievable by quantum theory and noncontextual models. We provide a more robust discussion of this result in Supplementary Notes 2, 3 and 4, where we model the noise-resistance of the contextual advantage result, describing also related loopholes for testing contextuality of the obtained data. In the remaining of the certification we will solely focus on 
nonclassicality provided by set-coherence.

Violations of the inequalities of Ref.~\cite{wagner2022coherence} are a  promising, scalable and  efficient way to witness coherence inside multi-mode interferometers, as described in Fig.~\ref{fig:high-dimension interference}b (see also  Supplementary Note 1). A multipath interferometer corresponds to an efficient device for generating high-dimensional coherent states and measuring their two-state overlaps. Consider any two states, $\vert \psi_1 \rangle = U(\boldsymbol{\theta}_1,\boldsymbol{\phi}_1)\vert 0\rangle $ and $\vert \psi_2 \rangle = U(\boldsymbol{\theta}_2,\boldsymbol{\phi}_2)\vert 0\rangle $ over some finite-dimensional Hilbert space in which $\ket{0}$ is one state of a given basis. Their overlap can be measured  choosing the two stages of a generic interferometer such that 
\begin{equation}
    r_{1,2}\equiv r_{2,1}=\vert \langle \psi_2 \vert \psi_1 \rangle \vert ^2 = \vert \langle 0 \vert U(\boldsymbol{\theta}_2,\boldsymbol{\phi}_2)^\dagger  U(\boldsymbol{\theta}_1,\boldsymbol{\phi}_1) \vert 0\rangle \vert^2.
\end{equation}
Using multi-mode devices it is possible to witness not only coherence, but \textit{coherence achievable only with qudits} by violation of the following family of inequalities defined recursively:

\begin{equation}\label{eq: family}
    h_{n}(r) = h_{n-1}(r)+r_{0,n-1}-\sum_{i=1}^{n-2}r_{i,n-1} \leq 1 ,
\end{equation}
where the sequence starts with $h_{3}(r) = r_{0,1}+r_{0,2}-r_{1,2}$ and the above equation defines inequalities for any integer $n>3$. When $n=4$ we have 
\begin{equation}    
h_4(r)=r_{0,1}+r_{0,2}+r_{0,3}-r_{1,2}-r_{1,3}-r_{2,3}\leq 1. 
\label{ineq: qutrit}
\end{equation}

The inequality in Eq.~\eqref{ineq: qutrit} cannot be violated by a set of pure qubit states.
With qutrits, it reaches violations up to $1/3$. Hence, quantum violations of inequality \eqref{ineq: qutrit} represent witnesses of both coherence and Hilbert space dimension higher than $2$. We have numerically found the same behaviour for sets of pure quantum states for the family of $h_n$ inequalities~\eqref{eq: family} up to $h_{10}$ by maximizing over parameters describing up to 10 states. This is evidence that the family of inequalities \eqref{eq: family} corresponds to simultaneous witnesses of coherence and dimension, achievable only by the device's capability of precisely preparing and measuring high-dimensional coherent states. We do not prove that these inequalities have this property for all values of $n$. Using semidefinite programming (SDP) techniques, we show that we can map the maximum possible values of the inequalities~\eqref{eq: family} to the solutions of a quadratic SDP. We then show that, for $n$ up to $n=2^{12}$, sets of  states with dimension $d=n-1$ are capable of violating the inequalities $h_n$, while no violation is obtained from states spanned by a basis of any lower dimension. In Supplementary Note 1 we present these theoretical results{, and discuss the underlying assumptions for the dimensionality certification} in detail.

{\subsection{Experimental Implementation and Results}}

Quantum interrogation and coherence witnesses are tested in heralded single-photon experiments, by means of the experimental setup shown in Fig.~\ref{fig:chip}, composed of a single-photon source based on parametric down-conversion, and a programmable integrated universal photonic processor (UPP) fabricated via femtosecond laser micro-machining (see Methods for more details).
Let us now discuss the results obtained in the performed experiments.

\vspace{1em}
\subsection*{Coherence and contextuality in two-level systems}

In the previous section, we have provided the theoretical framework which derives families of coherence witnesses based on the evaluation of pairwise overlaps among states in a finite set. Some inequalities tailored to work as coherence witnesses can also witness quantum contextuality. We have already mentioned inequality \eqref{ineq: 3-cycle inequalities} as one of such example. Furthermore, this inequality predicts an advantage for the efficiency in the task of quantum interrogation \cite{wagner2022coherence}. We implement and test this task in a {programmable} MZI inside the {integrated UPP}. The experimental optical circuit is shown Fig.~\ref{fig:expbomb_int}a. In our experiment, the absorbing object is modelled by a completely transparent {beam splitter} with reflectivity $r_{B}=\sin{\theta_B}=0$ placed in one arm of the MZI. 
The MZI is calibrated so that the two {beam splitters} have the same beam-splitting ratio ($\theta_1=\theta_2=\theta$), with a null internal phase ($\phi = 0$). These conditions guarantee that a {single photon} injected in mode 0 will always come out of output 0 when the object is absent.
The aim of this experiment is to estimate the efficiency $\eta$ of detecting the presence of the object without it absorbing the photon, as defined in Eq. \eqref{eq:eff_bomb}. In our scheme $p_{succ}$ corresponds to the fraction of single photons detected in mode 1, since this output is only possible when the object is present. The probability of absorption $p_{abs}$ is given by the fraction of photons detected in mode 2. In Fig. \ref{fig:expbomb_int}b we report the measurements of $\eta$ for different values of the reflectivity $r = \sin \theta$ of the two {beam splitters}. The theoretical curve is given by a  quantum model of the MZI, whose performance is, in general, not achievable by any generalized noncontextual model (see also Supplementary Note 2), given by

\begin{figure}[t]
    \centering
    \includegraphics[width=\columnwidth]{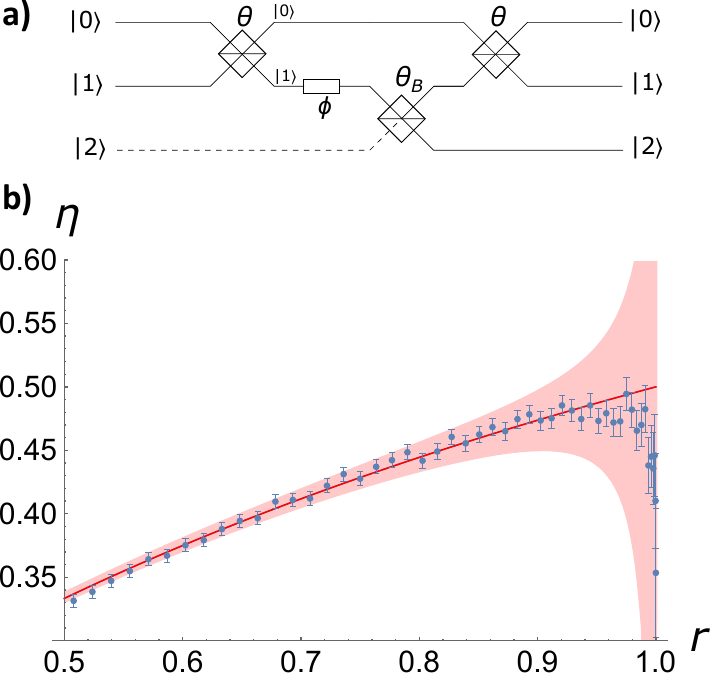}
    \caption{\textbf{Coherence and contextuality in a Mach-Zehnder interferometer.}
    a) Scheme of the {circuit} employed for the quantum interrogation task. b) Efficiency $\eta$ of the object's detection versus the reflectivity $r=\sin\theta$ of the two {beam splitters} in the MZI. The red curve is the theoretical prediction while the red shaded area represents the deviations from the ideal model due to dark counts and imperfect calibration of the {beam splitters}. The error bars derive from the poissonian statistics of the single-photon counts.}
    \label{fig:expbomb_int}
\end{figure}

\begin{figure*}[t]
    \centering
    \includegraphics[width=\textwidth]{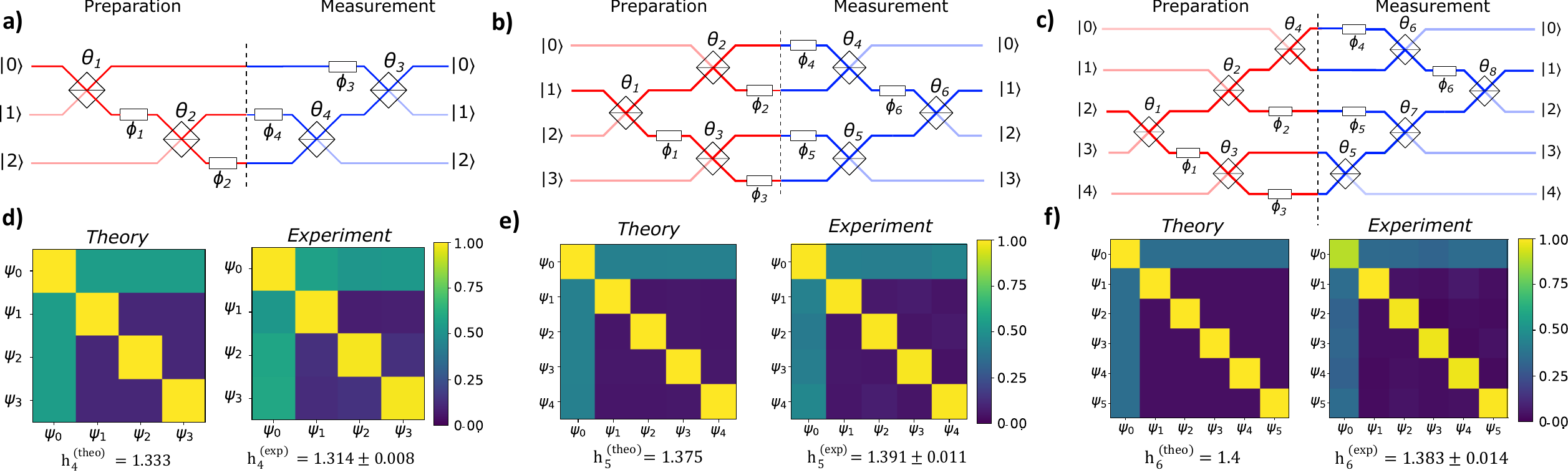}
    \caption{\textbf{Violations of 
    $h_n$ inequalities by qudits.}  a-c) Circuit schemes for qutrits a), ququarts b) and $5$-mode qudits c) preparation (red part) and measurement (blue). The single-photon signal enters from mode 0 for the qutrits case, from mode 1 for ququarts and from mode 2 for the 5-mode case. 
    d-f) Comparison of the theoretical and the experimental matrices of the pairwise overlaps values $r_{i,j}=\vert \langle \psi_j \vert \psi_i \rangle \vert ^2$ for the sets of states that maximize the violation of $h_4$, $h_5$, and $h_6$. 
    In particular, $h_4$ is violated by sets of 4 qutrit states $\{\psi_0, \dots \psi_3\}$ in d), the $h_5$ inequality is violated by sets of 5 ququart  states $\{\psi_0, \dots, \psi_4\}$ e) and $h_6$ is violated by sets of 6 quantum states of dimension 5 $\{\psi_0, \dots, \psi_5\}$ f). The uncertainty reported for each inequality derives from the poissonian statistics of the single-photon counts.}
    \label{fig:qudit}
\end{figure*}

\begin{equation}
   \eta= \frac{r (1-r)}{r (1-r)-r+1}
   \label{eq:bomb_noncontext}
\end{equation}
Our experimental data follows very well the predictions of the quantum model, showing not only that the device \textit{generates} data that cannot be explained with noncontextual models, but also that it \textit{uses} contextuality as a resource to achieve quantum-over-classical performance, quantified by the efficiency $\eta$. We observe that the largest deviations from the theoretical curve appear for values of $r$ close to 1. This discrepancy can be justified by taking into account the experimental imperfections in the apparatus (see Methods). 

In Supplementary Notes 2, 3 and 4 we present a detailed discussion about contextuality in  {MZIs}, including an analysis of the requirements for witnessing contextuality when the device is used for the quantum interrogation task. In there, we pick a specific {beam splitter} configuration for the interrogation task, and show that we experimentally achieve $\eta_{exp} = 0.428\pm 0.006$, while noncontextual models must have an efficiency $\eta^{NC} \le 0.410$, even when benefiting from the effect of noise, which raises the noncontextual upper bound for $\eta$.

In Supplementary Note 5 we also certify coherence in the MZI using a different, and novel 
inequality featuring a high level of violation by five symmetrical states on a great circle of the Bloch sphere.

\subsection*{Coherence and dimension witnesses in higher dimensions}

Eq.~\eqref{eq: family} describes a family of inequalities that are tailored for certifying coherence in systems with dimension $d>2$. We will refer to such inequalities as $h_n$, where $n$ is the number of states in the set whose overlaps we need to evaluate. They arise as inequalities obtained using the event graph approach \cite{Wagner2022},  when we consider complete graphs $K_n$. The presence of coherence in the states is witnessed when an inequality is violated, that is, when $h_n >1$. An important point to note is that the values $h_n(r)$ provide information regarding the coherence accessible only due to the dimension of the space. In fact,  the $h_n$ inequalities are not violated by sets of states without coherence, nor by systems with dimension $d<n-1$. For example, one-qubit states do not violate $h_4\leq 1$,  qutrit states do not violate $h_5\leq 1$, and so on. 
The main result is that $h_n$ displays different maximum values according to the dimension of the system. This implies that the functionals $h_n(r)$ are not only dimension witnesses but also \textit{indicators} of the dimension of the space. 


\begin{figure*}[t]
    \centering
    \includegraphics[width=\textwidth]{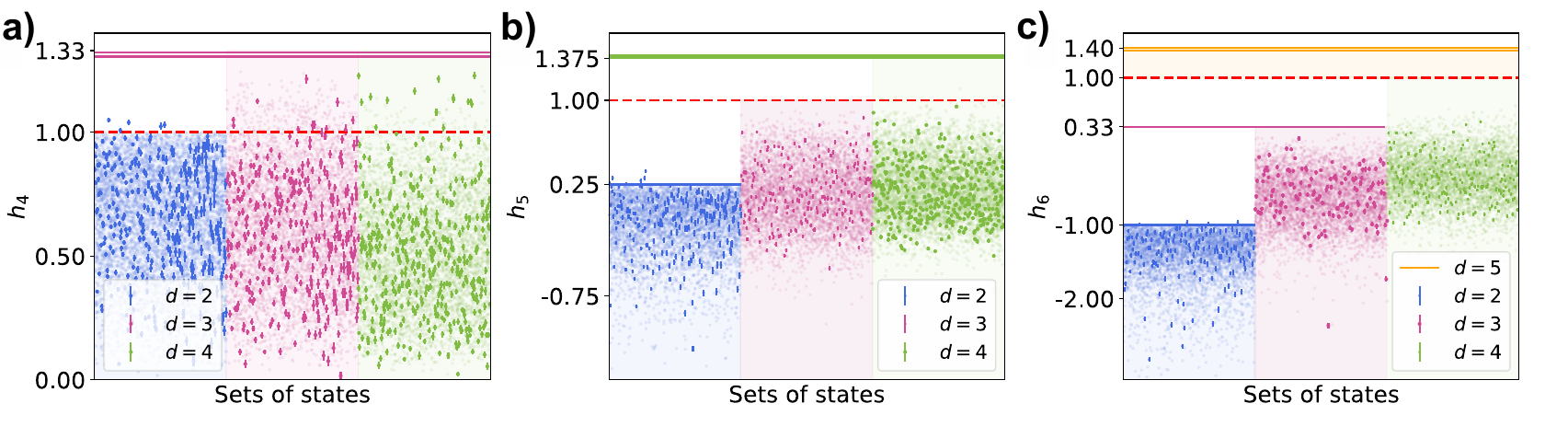}
    \caption{\textbf{$h_n$ inequalities as dimension witnesses.} Random states sampled uniformly in Hilbert spaces of dimensions 2 (blue), 3 (purple) and 4 (green). Bold points correspond to $\sim 200$ experimental preparations of sets of uniformly random qudit states for each dimension. The uncertainties are smaller than the points' size. The shaded points in the background are $\sim 5000$ sets of numerically simulated states for each dimension. The dotted red line indicates the threshold value $1$ for $h_n$ required to witness coherence. a) Distributions of the $h_4(r)$ values. The purple line is the maximum violation of the $h_4$ inequality for state dimension $d=3$. b) Same analysis for the $h_5$ inequalities. Here, we observe a lower bound for $d=2$, highlighted by the blue line, that allows better discrimination of high-dimensional sets of states. The green line is the maximum violation achieved with the set which includes states of $d=4$. c) Distributions for the $h_6$ inequality. We report only the maximum violation measured for the coherence witness since our setup is not a universal state preparator for $d=5$. Blue and purple lines highlight the maximum values of $h_6$ for $d=2$ and $d=3$ respectively.}
    \label{fig:dim_witness}
\end{figure*}

We tested firstly the effectiveness of $h_4$, $h_5$, and $h_6$ inequalities as coherence witnesses. In Fig. \ref{fig:qudit} we show the circuits to prepare and measure 3-, 4- and 5-mode qudits. 
In particular, the red part of the circuits for qutrit (Fig. \ref{fig:qudit}a) and ququart (Fig. \ref{fig:qudit}b) are  universal state preparators when a single photon enters the device respectively in inputs 0 and 1. In the case of 5-mode qudits, the 6-mode {UPP} does not have enough layers of MZIs to implement  independent universal preparation and measurement stations. However, the circuit in Fig. \ref{fig:qudit}c can prepare a set of six 5-mode qudit states that maximize $h_6$. In the figure, we also report the $h_n$ values together with the matrix of the pairwise overlaps. We estimated the violations by considering only the upper triangular part of such a matrix, i.e $r_{i,j}$ with $i>j$. 
Further details on the sets of states used to violate the inequalities, and a discussion on the sources of noise in the experimental measurements are reported in Supplementary Notes 6 and 7.

We then move to the experimental test of $h_n$ as dimension witnesses by sampling uniformly random states spanned by bases of different dimensions. The distributions of the values obtained for the l.h.s. of the $h_n$ inequalities for $d=2$, $d=3$ and  $d=4$ are reported in Fig. \ref{fig:dim_witness}, together with the maximum theoretical values of $h_n$ for systems of those dimensions. The experimental data confirm the theoretical predictions. This provides further insight on the power of the $h_n$ as dimension witnesses.  In fact, the uniform sampling of the states tries to answer the following question: how much information about the dimension of the space can we retrieve from the value of the l.h.s. of the $h_n$ inequalities, without knowing the optimal set that maximizes the violation?
 We did not sample random states in $d=5$ since our circuits are not universal state preparators and measurement devices for qudits of this dimension. Table \ref{tab:exp_dim} summarizes the maximum values of the functionals $h_n$ obtained 
 in the random sampling and in the previous analysis dedicated to coherence witnesses. We see that the maximum violations are not typically achievable when sampling random states. However, the $h_n$ become very effective in discerning systems of $d>2$ for increasing values of $n$.

\begin{table}[t]
    \centering
    \begin{tabular}{p{0.5cm} p{1.75cm} p{1.7cm} p{1.7cm} p{1.7cm}}
    \toprule
      & \textbf{$d=2$} & \textbf{$d=3$} & \textbf{$d=4$} & \textbf{$d=5$}\\
     \midrule
     \textbf{$h_4$} & $1.05\pm0.01$ & $\mathbf{1.31 \pm 0.01}$ & $1.23\pm0.01$ & \\ 
     \textbf{$h_5$} & $0.36\pm0.02$ & $0.84\pm0.02$ &  $\mathbf{1.39 \pm 0.01}$ &   \\ 
    \textbf{$h_6$} & $-0.96 \pm0.04$ & $0.17\pm 0.03$ & $0.40\pm 0.02$ & $\mathbf{1.38\pm 0.01}$  \\
    \bottomrule
\end{tabular}
\caption{\textbf{Experimental results for witnesses of coherence and dimension.} In bold the $h_n$ values that we measured for the coherence witnesses, and in Roman the maximum experimental values measured in the uniform random sampling of $\sim 200$ sets of states of dimension $d$.}
\label{tab:exp_dim}
\end{table}

In summary, we showed how to exploit new families of overlap inequalities to witness coherence in qudit systems. The coherence is certified when $h_n>1$. Furthermore, the $h_n$ inequalities introduced in Ref.~\cite{Wagner2022} and tested here are only violated by systems having both coherence and a dimension $d\ge n-1$. Even when $h_n<1$, while not witnessing coherence, the value of $h_n$ still provides information on the dimension.


\section*{Discussion}

In this work, we have characterized how quantum information is processed within {a six-mode programmable integrated UPP}. To do so, we witnessed - and used - two different notions of quantum nonclassicality, namely generalized contextuality, and coherence. 
Our characterization of nonclassicality is done in a way that depends on the dimension of the Hilbert space generated by single-photon interference through the paths of the programmable device. Our analysis begins with the simplest scenario, where we use a subsection of the device to implement a two-mode Mach-Zehnder interferometer {MZI}. We demonstrate the presence of generalized contextuality within the MZI by violation of a recently introduced novel generalized noncontextuality inequality. We show that this resource is used to achieve efficiencies in the task of quantum interrogation that are higher than those possible by any noncontextual model that reproduces the same operational constraints considered in our experiment.

Then we proceeded to investigate nonclassicality generated by a single photon able to propagate through gradually larger portions of the interferometer. In doing so, we introduce a novel theoretical perspective in coherence theory: quantum coherence achievable only by qudits. We show that a family of inequalities is capable of identifying coherence that can only be witnessed when the totality of the Hilbert space dimension considered is used in non-trivial ways. We experimentally measure the presence of this kind of coherence for a single photon interfering in up to 5 modes. Via numerical simulations, we demonstrate that this family of inequalities is a simultaneous witness of coherence and Hilbert space dimension $d$ up to $d = 2^{12}$. 

{Our certification scheme leaves some opportunities for further theoretical investigation. For instance, while our scheme is not device independent, we believe that in the future it may be suitable for a description in the semi-device independent framework, since our single requirement over the data is that it corresponds to two-state overlaps, possibly interpreted as a promise over the possible measurements. Quantum states, Hilbert spaces and physical devices are arbitrary. Also, due to the invariant properties of the inequalities, their maximization will likely be related to the task of self-testing~\cite{navascues2023selftesting}, and techniques used there might be applicable in our case.}

We believe that violation of these inequalities can be exploited in the future as a novel certification technique benchmarking nontrivial high-dimensional coherence, and that may be related to  hardness of quantum computation. Moreover, the theoretical results presented here apply to any platform for quantum computation, and not just photonics. 

\section*{Methods}

\subsection*{Experimental setup}
\noindent  

A BBO crystal is pumped by a pulsed-laser at the wavelength of $\lambda = 392.5$ nm. The spontaneous parametric down conversion (SPDC) process generates photon pairs at $\lambda = 785$ nm.  In the experiment we focus on one pair emissions, much more likely to happen than multi-pair generation. 

One of the two photons is employed as a trigger signal {while the other one is injected in a universal photonic processor {(UPP), i.e. a fully-programmable multi-mode interferometer.} This device consists in a {waveguide circuit}, fabricated in-house by the femtosecond laser writing technology in an alumino-borosilicate glass chip.  }
The scheme of the {photonic circuit of the {six-mode} {UPP} is reported in Fig.~\ref{fig:chip}c and  follows the decomposition into a {rectangular network of {beam splitters} and phase shifters devised in \cite{Clements2016}.} 
{The beam splitters in the scheme are actually MZIs (see inset in Fig.~\ref{fig:chip}c), which provide variable splitting ratios depending on the value of the internal phase. Dynamic reconfiguration of the {UPP} operation is accomplished by thermo-optic phase shifters, which enables the active control of the values of the phase terms placed inside and outside the cascaded MZIs. The input and output ports of the waveguide circuit are optically connected via fiber arrays (single mode fibers at the input, multi-mode fibers at the output, see also Fig.~\ref{fig:chip}b).
 
In particular, the phase shifters are based on {gold microheaters}, deposited and patterned on the chip surface. Upon driving suitable currents into the {microheaters}, local heating of the substrate is achieved in a precise and controlled way. Such local heating induces in turn a refractive index change and thus controlled phase delays in the waveguides due to the thermo-optic effect.  A careful calibration of the phase shifters allows to implement with our {UPP} any linear unitary transformation of the input optical modes. Such calibration is performed by classical coherent light and does not rely on the quantum theory we test in our experiments. More details on the design, fabrication and calibration process of the {UPP} are provided in Supplementary Note 8.

Finally, t}he outputs and the trigger photon fibers are connected to avalanche-photodiode single-photon detectors (APDs). The detector signals are processed by a time-to-digital converter (TDC) that counts the two-fold coincidence between the chip outputs and the trigger photon. 

\subsection*{Noise model for the quantum interrogation experiment}
The main sources of noise that need to be considered for the quantum interrogation experiment are mismatches in the reflectivity $r$ of the two {beam splitters}, as well as dark counts of the detectors. These become more significant for $r\sim 1$, since in this regime both $p_{succ}$ and $p_{abs}$ tend to be very small.
Our noise-corrected formula for $\eta$  will be
\begin{equation}
   \eta_{noisy}= \frac{r (1-(1\pm\varepsilon)r)+n_1}{r (1-(1\pm\varepsilon)r)-r+1+n_1+n_2},
\end{equation}
where $\varepsilon$ is the percentage mismatch of the two reflectivities and $n_1$ and $n_2$ indicate the ratio between dark counts and signal.

The red area in Fig. \ref{fig:expbomb_int}b encloses the set of curves resulting from a span of the parameters $\varepsilon$, $n_1$ and $n_2$ in the range 0 and 0.005. Our noise model predicts large deviations from the ideal quantum efficiency when the beam-splitting ratio approaches $r=1$.

\subsection*{Pairwise overlap inequalities characterizing incoherent sets of states}

We will briefly recall the arguments from Refs.~\cite{galvaobrod2020quantum,wagner2022coherence} for why the selected inequalities are capable of bounding the overlaps of sets of incoherent states, and how a graph-theoretic construction enables finding such inequalities. 

For a diagonal set of states, with respect to \textit{any} basis, the two-state overlaps $\text{Tr}(\rho_i\rho_j)$ of the elements $\{\rho_i\}_i$ of that set represent the probability of obtaining equal outcomes from the states upon measurements with respect to the reference basis. When such an interpretation is possible, we say that the set of states is coherence-free, or incoherent, and the reference basis is a coherence-free model for this set. For each set of states $\underline{\rho}$ it is possible to define an edge-weighted graph $(G,r)$ with vertices of the graph $V(G)$ representing quantum states and edges $E(G)$ having weights $r_e \equiv r_{i,j} := \text{Tr}(\rho_i\rho_j)$. If we collect all weights into a tuple $r = (r_e)_e \in \mathbb{R}^{|E(G)|}$ with $|E(G)|$ the total number of edges in the (finite) graph $G$, it is possible to bound all tuples $r$ resulting from states that are diagonal with respect to \textit{some} basis. This was studied in Ref.~\cite{galvaobrod2020quantum,Wagner2022}, and the bounds are given by linear inequalities. Similar to what is already well established in Bell nonlocality~\cite{brunner2014nonlocality} and contextuality~\cite{budroni2021kochen}, inequalities bound descriptions of classicality; those inequalities are built to bound coherence-free models for $\underline{\rho}$.

The basic reasoning described in its most general form starts by acknowledging that, if all states are incoherent, there exists some set of output outcomes $\Omega$ with respect to which  the weights $r_{i,j}$ represent the probability that, upon independently measuring $\rho_i,\rho_j$ from adjacent nodes $i,j$ in the graph $G$ we obtain equal outcomes. As an example, consider two adjacent nodes described by the maximally mixed qubit state $I_2/2$. In this case the edge-weight corresponds to the two-state overlap $\text{Tr}(I_2/4) = 1/2$. This is also the probability that we measure these two states with respect to the basis that diagonalizes them, and obtain equal outcomes, i.e., the probability that two ideal coins return equal outcomes.

To find overlap inequalities the algorithm then goes as follows: for a given graph $G$ we have all conceivable tuples $r = (r_e)_e$ described by all assignments $0$ or $1$. Any tuple of two-state overlaps will be inside the polytope described by the hypercube $[0,1]^{|E(G)|}$. Using the assumption of incoherent states and what this forces overlaps to satisfy, i.e., to be the probability of equal outcomes with respect to some set of labels, one can \textit{forbid} some assignments from the hypercube. The remaining ones are just those possible from an incoherent interpretation of states and edge-weights. In convex geometry, the convex hull of this set of assignments defines a polytope, and using standard tools it is possible to find the facet-defining inequalities for this polytope, given that the vertices are known. These facet-defining inequalities are the inequalities we probe in our work. By construction, inequality violations immediately contradict the hypothesis of incoherent states, hence serving as witnesses of set-coherence. 

In this work, the graphs considered are only complete graphs $K_n$, a graph where every is connected to every other node, for $n=3,4,5,6$. The label $n$ describes the number of nodes for the graph. The inequalities described by Eq.~\eqref{eq: family} form one among many inequalities that can be derived for the graph $K_n$, having the particularly interesting properties we discuss in the main text. In the Supplemental Material 
we describe a different inequality from the same graph that has a large violation that can be violated with qubits.  For a more detailed description we refer to Ref.~\cite{Wagner2022}.

\subsection*{Standard quantum interrogation task}

For the interrogation task one assumes that there might be some object in one of the interferometer's arms, depicted as a question mark in Fig.~\ref{fig:MZI}a. The object, if present, is assumed to completely absorb incoming photons, and the task is to detect the presence of the object without any photons being absorbed by the object. It is somewhat surprising that such a task can be accomplished at all, but using the fact that {beam splitters} create coherence, a simple set-up is capable of performing the task. 
Letting two {50:50} beam splitters and no difference in phase between the arms, in case there is no object all the coherence created in the preparation stage is destroyed in the measurement stage, and one of the detectors never lights up. However, in the presence of an object, it acts as a complete path-information measurement device inside the interferometer, projecting the state of the photon inside the interferometer to the arm where the object is \textit{not} present. This happens with 50\% probability. When the non-absorbed photon hits the second {beam splitter} in the measurement stage there is a 50\% chance that the detector that would be dark in the absence of an object now lights up. Therefore, with probability 25\% one can detect the presence of the object without directly interacting with it. Surprisingly, there are noncontextual models capable of reproducing precisely this feature~\cite{catani2021whyinterference}.

We can vary the beam splitting ratios for the protocol, which allows us to have an improvement in the efficiency of the task. It is crucial that the second beam splitter in the MZI perfectly reverses the first beam splitter's action, so that whenever we have no object inside the interferometer, or we have an inactive object, one of the detectors never clicks (remaining dark). The probabilities $p_{succ}$ and $p_{abs}$ depend on the beam-splitting ratio characterizing the beam splitter. The probability $p_{abs}$ is simply the probability that once the photon enters the device it goes to the arm that has the object. The probability $p_{succ}$ is the probability that the photon does not reach the object, hence going to the other arm, and leaves the MZI via the output port that is dark in the absence of the object.

\bibliography{bibliography}

\begin{acknowledgments}
We would like to thank Emmanuel Zambrini, Vinicius P. Rossi{, Amit Te'eni} and Antonio Ruiz-Molero for fruitful discussions. {The programmable UPP was partially fabricated at PoliFAB, the micro- and nanofabrication facility of Politecnico di Milano (\href{https://www.polifab.polimi.it/}{https://www.polifab.polimi.it/}). C.P., F.C. and R.O. wish to thank the PoliFAB staff for the valuable technical support.}

{\textbf{Funding:}} R.W. and A.Ca. acknowledge support from FCT -- Fundação para a Ciência e a Tecnologia (Portugal) through PhD Grants SFRH/BD/151199/2021 and SFRH/BD/151190/2021. E.F.G. acknowledges support from FCT – Fundação para a Ciência e a Tecnologia (Portugal) via project CEECINST/00062/2018. {We acknowledge
support from the ERC Advanced Grant QU-BOSS (QUantum advantage via nonlinear BOSon Sampling, grant agreement no. 884676) and from PNRR MUR project PE0000023-NQSTI (Spoke 4 and Spoke 7)}. This work is also supported by Horizon Europe project FoQaCiA, GA no.101070558. 

\textbf{Author contributions:} 
T. G., C. E., F. H., G. C., N. S., F.S., R. W,   A. Ca., E. F. G. conceived the concept and experiment.  R. W., A. Ca., and E. F. G. developed the theory of the method. 
C. P., F. C, S.P., A. Cr., and R.O. fabricated the photonic chips and characterized the integrated device using classical optics.  C. E., T. G., F. H., G.C., N.S., and F.S. carried out the quantum experiments and performed the data analysis. 
All the authors discussed the results and contributed to the writing of the paper.

\textbf{Competing Interests:}
{F.C. and R.O. are co-founders of the company Ephos. The authors declare no other competing interests.}

\textbf{Data and materials availability:}
{All data needed to evaluate the conclusions in the paper are present in the paper and/or the Supplementary Materials or in the repository \href{https://zenodo.org/record/8386354}{https://zenodo.org/record/8386354}.}

\end{acknowledgments}

\end{document}


\title{Supplemental Material for: Experimental certification of contextuality, coherence and dimension in a programmable {universal photonic processor}                                                                                                                                                                             }

\author{Taira Giordani} 
\thanks{These two authors contributed equally}
\affiliation{Dipartimento di Fisica, Sapienza Universit\`{a} di Roma, Piazzale Aldo Moro 5, I-00185 Roma, Italy}

\author{Rafael Wagner}
\thanks{These two authors contributed equally}
\affiliation{International Iberian Nanotechnology Laboratory (INL)
 Av. Mestre José Veiga s/n, 4715-330 Braga, Portugal
}
\affiliation{Centro de F\'{i}sica, Universidade do Minho, Campus de Gualtar, 4710-057 Braga, Portugal}

\author{Chiara Esposito} 
\affiliation{Dipartimento di Fisica, Sapienza Universit\`{a} di Roma, Piazzale Aldo Moro 5, I-00185 Roma, Italy} 

\author{Anita Camillini}
\affiliation{International Iberian Nanotechnology Laboratory (INL)
 Av. Mestre José Veiga s/n, 4715-330 Braga, Portugal
}
\affiliation{Centro de F\'{i}sica, Universidade do Minho, Campus de Gualtar, 4710-057 Braga, Portugal}

\author{Francesco Hoch} 
\affiliation{Dipartimento di Fisica, Sapienza Universit\`{a} di Roma, Piazzale Aldo Moro 5, I-00185 Roma, Italy} 

\author{Gonzalo Carvacho} \affiliation{Dipartimento di Fisica, Sapienza Universit\`{a} di Roma, Piazzale Aldo Moro 5, I-00185 Roma, Italy}

\author{Ciro Pentangelo}
\affiliation{Dipartimento di Fisica, Politecnico di Milano, Piazza Leonardo da Vinci, 32, I-20133 Milano, Italy}
\affiliation{Istituto di Fotonica e Nanotecnologie, Consiglio Nazionale delle Ricerche (IFN-CNR), 
Piazza Leonardo da Vinci, 32, I-20133 Milano, Italy}

\author{Francesco Ceccarelli}
\affiliation{Istituto di Fotonica e Nanotecnologie, Consiglio Nazionale delle Ricerche (IFN-CNR), 
Piazza Leonardo da Vinci, 32, I-20133 Milano, Italy}

\author{Simone Piacentini}
\affiliation{Istituto di Fotonica e Nanotecnologie, Consiglio Nazionale delle Ricerche (IFN-CNR), 
Piazza Leonardo da Vinci, 32, I-20133 Milano, Italy}

\author{Andrea Crespi}
\affiliation{Dipartimento di Fisica, Politecnico di Milano, Piazza Leonardo da Vinci, 32, I-20133 Milano, Italy}
\affiliation{Istituto di Fotonica e Nanotecnologie, Consiglio Nazionale delle Ricerche (IFN-CNR), 
Piazza Leonardo da Vinci, 32, I-20133 Milano, Italy}

\author{Nicol\`o Spagnolo} \affiliation{Dipartimento di Fisica, Sapienza Universit\`{a} di Roma, Piazzale Aldo Moro 5, I-00185 Roma, Italy}
\author{Roberto Osellame}
\affiliation{Istituto di Fotonica e Nanotecnologie, Consiglio Nazionale delle Ricerche (IFN-CNR), 
Piazza Leonardo da Vinci, 32, I-20133 Milano, Italy}


\author{Ernesto F. Galv\~ao}
\email[Corresponding author: ]{ernesto.galvao@inl.int}
\affiliation{International Iberian Nanotechnology Laboratory (INL)
 Av. Mestre José Veiga s/n, 4715-330 Braga, Portugal
}

\affiliation{Instituto de F\'isica, Universidade Federal Fluminense, Av. Gal. Milton Tavares de Souza s/n, Niter\'oi, RJ, 24210-340, Brazil}

\author{Fabio Sciarrino}
\email[Corresponding author: ]{fabio.sciarrino@uniroma1.it}
\affiliation{Dipartimento di Fisica, Sapienza Universit\`{a} di Roma, Piazzale Aldo Moro 5, I-00185 Roma, Italy}

\maketitle

\section{Simultaneous witness of dimension and coherence}\label{appendix: dim_cohe}
\begin{figure}[b]
    \centering
    \includegraphics[width = 0.7 \textwidth]{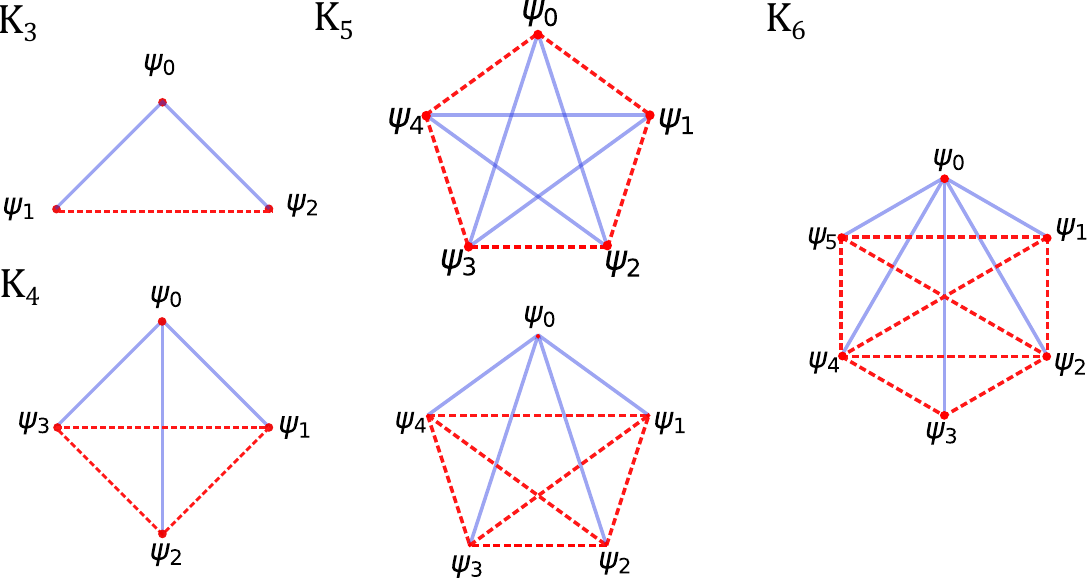}
    \caption{\textbf{Examples of the graphs $K_n$ representation of the inequalities, where $n$ are the number of vertex.} The vertex represents the quantum states and the edges the pairwise overlaps. The inequalities are expressed through the number of edges  weighted with $+1$ (blue lines) and $-1$ (red dotted lines). For example, the $K_5$ graph represents the $h_{MZI}$ in the configuration on the top or the $h_5$ in the bottom one.}
    \label{fig: graphs}
\end{figure}

Here we argue theoretically and numerically that the family of inequalities described in the main text on Eq.~(4), for $3<n<2^{12}$ correspond to witnesses of both coherence and dimension for sets of quantum states. From now on, we will count states representing nodes of  $K_n$ as $1,\dots,n$ instead of $0,\dots,n-1$. See also Fig. \ref{fig: graphs}. We first prove rigorously that pure state qubits cannot violate the $h_4$ inequality; showing that this inequality is both a dimension and coherence witness. 

\begin{theorem}
Let $\{\vert \psi_i \rangle \}_{i=1}^4$ be any 4-tuple of pure quantum states of dimension $2$. Then, letting $r_{i,j} \equiv \vert \langle \psi_i \vert \psi_j \rangle \vert^2$ it is impossible to violate the inequality $h_4(r) = r_{1,2}+r_{1,3}+r_{1,4}-r_{2,3}-r_{2,4}-r_{3,4} \leq 1$. 
\end{theorem}

\begin{proof}
    We want to maximize $h_4$ for qubit  states. The general form for this maximization procedure for any possible state is given by,
    \begin{align*}    \max_{\rho_1,\rho_2,\rho_3,\rho_4 \in \mathcal{D}(\mathcal{H})}\left(\text{Tr}(\rho_1\rho_2)+\text{Tr}(\rho_1\rho_3)+\text{Tr}(\rho_1\rho_4)-\text{Tr}(\rho_2\rho_3)-\text{Tr}(\rho_2\rho_4)-\text{Tr}(\rho_3\rho_4)\right),
    \end{align*}
where $\mathcal{D}(\mathcal{H})$ is the set of all density matrices over $\mathcal{H}\simeq \mathbb{C}^2$. In fact, if we consider only the maximization with respect to $\rho_1$ first, we see that $\rho_1$ appears only in the first 3 overlaps. This allows us to use the following relation,

\begin{equation}
    \text{Tr}(\rho_1\rho_2)+\text{Tr}(\rho_1\rho_3)+\text{Tr}(\rho_1\rho_4) = \text{Tr}\left(\rho_1\left( \sum_{i=2}^4 \rho_i \right)\right) \leq \left \Vert \sum_{i=2}^4 \rho_i \right \Vert
\end{equation}
    where in the first equality we have used linearity of the trace and for the second inequality, we define $\Vert \cdot \Vert$ is the operator norm,
    \begin{equation}
        \Vert A\Vert := \sup_{v \in \mathcal{H}, v\neq 0} \frac{\Vert A v \Vert_{\mathcal{H}} }{\Vert v \Vert_{\mathcal{H}} }
    \end{equation}
    where $\Vert \cdot \Vert_{\mathcal{H}}$ is the norm that makes $\mathcal{H}$ a Hilbert space, i.e., the norm arising from the inner product.   In particular, because the sum of positive semidefinite matrices is again positive semidefinite $\sum_i \rho_i$ is positive semidefinite, and the inequality is tight meaning that there is a state $\rho_1$ for any given $\rho_2,\rho_3,\rho_4$ such that the equality holds. Therefore we end up translating our problem into  

    \begin{equation}
        \max_{\rho_2,\rho_3,\rho_3\in \mathcal{D}(\mathcal{H})}\left(\Vert \rho_2+\rho_3+\rho_4\Vert-\text{Tr}(\rho_2\rho_3)-\text{Tr}(\rho_2\rho_4)-\text{Tr}(\rho_3\rho_4))\right),
    \end{equation}
and we proceed to study the quantity $\Vert \rho_2+\rho_3+\rho_4\Vert$. Due to the invariant nature of the overlap scenarios we might choose $\rho_2,\rho_3,\rho_4$ to be the density matrices related to the pure states $\vert 0 \rangle, \vert \theta \rangle, \vert \alpha,\varphi \rangle$ defined by,
\begin{equation}
    \vert \theta \rangle = \cos\theta\vert 0\rangle + \sin\theta \vert 1 \rangle, 
\end{equation}
and
\begin{equation}
    \vert \alpha, \varphi \rangle = \cos\alpha \vert 0 \rangle + e^{i\varphi}\sin\alpha \vert 1\rangle, 
\end{equation}
with no loss of generality. This implies that we will have a relation dependent only on $3$ parameters to investigate $\theta, \alpha \in [0,\pi/2]$ and $\varphi \in [0,2\pi]$.

Recall the following theorem,

\begin{theorem}
    Let $\mathcal{H}$ be a Hilbert space. Then, for any operator $A \in \mathcal{B}(\mathcal{H})$, where $\mathcal{B}(\mathcal{H})$ is the set of bounded operators with domain and codomain in the space, then we have that,
    \begin{equation}
        \Vert A \Vert = \sqrt{\max \sigma(A^*A)}
    \end{equation}
    where $\Vert \cdot \Vert$ is the operator norm, and $\sigma(X)$ is the spectrum of operator $X$.

\end{theorem}

The above result implies that to calculate the operator norm we simply need to find the maximal eigenvalue of the operator sum $\vert 0 \rangle \langle 0 \vert + \vert \theta \rangle \langle \theta \vert + \vert \alpha, \varphi \rangle \langle \alpha, \varphi \vert$. This sum describes the matrix,
\begin{equation}
    \left(\begin{matrix}
    1+\cos^2\theta+\cos^2\alpha & \frac{1}{2}(\sin(2\theta)+e^{i\varphi}\sin(2\alpha))\\
    \frac{1}{2}(\sin(2\theta)+e^{-i\varphi}\sin(2\alpha)) & \sin^2\theta + \sin^2\alpha 
    \end{matrix}\right)
\end{equation}

that has the following eigenvalues,



\begin{equation}\label{eq: maximum eigenvalue}
    \lambda_{\pm} = \frac{3}{2}\pm\frac{1}{2} \sqrt{2\sin (2 \alpha ) \sin (2 \theta
   )\cos(\varphi)+4 \cos (2 \alpha )
   \cos ^2(\theta )+2  \cos (2
   \theta )+3}
\end{equation}

The maximum eigenvalue is  given by  $\lambda_+$. After straightforward calculations, asking if $h_4(r)>1$ corresponds to asking if the following quantity can be larger than 0,

\begin{eqnarray*}
    &g(\theta,\alpha,\varphi) = \frac{3}{2}+\frac{1}{2} \sqrt{2\sin (2 \alpha ) \sin (2 \theta
   )\cos(\varphi)+4 \cos (2 \alpha )
   \cos ^2(\theta )+2  \cos (2
   \theta )+3} -1 - \cos^2(\theta)-\cos^2(\alpha)\\
   &-\cos^2(\theta)\cos^2(\alpha) -\sin^2(\theta)\sin^2(\alpha)-\frac{1}{2}\sin(2\theta)\sin(2\alpha)\cos(\varphi).
\end{eqnarray*}
The above function $g(\theta,\alpha,\varphi)$ corresponds to the optimization of general $\rho_1$ as done before for the functional $h_4(r)-1$ that we want to show to be less than or equal to zero. We may note that the above equation can be understood as a specific parametrization of the following function over $3$ variables, that we call $x,y,z$;

\begin{equation}
    f(x,y,z) = \frac{3}{2}+\frac{1}{2}\sqrt{2x+4y(2z-1)+2(2y-1)+3}-1-y-z-yz-(1-y)(1-z)-x
\end{equation}
where the specific parametrization arises from letting the variables become dependent from one-another in the specific form,
\begin{align*}
    x &= \frac{1}{2}\sin(2\theta)\sin(2\alpha)\cos(\varphi)\\
    y &= \cos^2(\theta)\\
    z &= \cos^2(\alpha).
\end{align*}

This implies that the maximum of the function $f(x,y,z)$, over the domain $[0,1]^3$, is, at most, larger than the maximum of the function $g(\theta,\alpha,\varphi) = f(x,y,z)\vert_{x=x(\theta,\alpha,\varphi), y = y(\theta), z = z(\alpha)} $. Hence, we may study the maximum value of $f$ directly to upper-bound $g$. We can re-write $f$ simplifying its form,

\begin{equation*}
    f(x,y,z) = -\frac{1}{2}+\frac{1}{2}\sqrt{2x+8yz+1}-2yz-x= \frac{1}{2}\sqrt{1+8yz+2x}-\frac{1}{2}(1+4yz+2x),
\end{equation*}

Now, we see that this function is always $\leq 0$ for any $x,y,z \in [0,1]$. We may re-write the function $f$ as,

\begin{equation}
    f(x,y,z) = \frac{1}{2}\sqrt{1+8yz+2x}-\frac{1}{2}(1+8yz+2x)+2yz,
\end{equation}
and now simply note that letting $c = 1+2x$ and $yz=\omega$ we have that the function $\frac{1}{2}(\sqrt{8\omega + c} - 8\omega - c)+2\omega$ is always less than, or equal to zero for $1\leq c \leq 3 $ and $0 \leq \omega \leq 1$. Equality holds when $c=1$ and $\omega = 0$. Therefore $f(x,y,z) \leq 0, \forall x,y,z \in [0,1]^3$ meaning that $g(\theta,\alpha,\varphi)\leq 0, \forall (\theta,\alpha,\varphi) \in [0,\pi/2]^2\times [0,2\pi)$ implying that $h_4(r)-1 \leq 0$ for any pure qubit realization of $r\equiv(r_{12},r_{13},r_{14},r_{23},r_{24},r_{34})$.

We conclude that pure qubit states cannot violate the $h_4$ inequality. 
\end{proof}

The above theorem shows that the $h_4$ inequality is a witness of both coherence and dimension. We now proceed to numerically investigate if the same property holds for the family of inequalities $h_n$ described in the main text. The results are shown in table \ref{tab:kn} below.

\begin{table}[ht]
    \centering
    \begin{tabular}{|*{10}{c|}}
    \cline{2-10}
    \multicolumn{1}{c|}{} & \multicolumn{9}{c|}{Dimension} \\
     \cline{2-10}
     \multicolumn{1}{c|}{} & \textbf{2} & \textbf{3} & \textbf{4} & \textbf{5} & \textbf{6} & \textbf{7} & \textbf{8} & \textbf{9} & \textbf{10} \\ 
     \hline
     $h_3$ & 1.250 & 1.250 & & & & & & &  \\ 
     $h_4$ & 1.000 & 1.333 & 1.333 & & & & & & \\ 
     $h_5$ & 0.250 & 1.000 & 1.375 & 1.375 & & & & & \\ 
    $h_6$ & -0.999 & 0.333 & 1.000 & 1.400 & 1.400 & & & & \\
    $h_7$ & -2.750 & -0.667 & 0.375 & 1.000 & 1.417 & 1.417 & & & \\
     $h_8$ & -5.000 & -2.000 & -0.500 & 0.400 & 1.000 & 1.429 & 1.417 & & \\
     $h_9$ & -7.750 & -3.667 & -1.625 & -0.400 & 0.417 & 1.000 & 1.428 & 1.429 & \\
     $h_{10}$ & -11.000 & -5.667 & -3.000 & -1.400 & -0.333 & 0.429 & 1.000 & 1.443 & 1.437\\ 
     \hline
\end{tabular}
 \caption{\textbf{Maximal values for $h_n$ inequality functionals from $K_n$ graphs.} Letting $h_n(r)$ be the functionals, bounding incoherent models for $h_n(r)\leq 1$, we numerically investigate the maximal values of $h_n(r)$ for $n=3,\dots,10$ the number of states  and $d=2,\dots,10$ the dimension of the Hilbert space of all the states.  The violations obtained for each inequality should not decrease as we increase the dimension, some examples of this in the table result from inefficiencies of the NMaximize \textit{Mathematica} function, e.g. in violations of $h_8$ for $d=7,8$.}
    \label{tab:kn}
\end{table}

This numerical investigation was performed by maximising the value $h_n(r)$ achieved for $r = (\vert \langle \psi_i \vert \psi_j \rangle \vert^2)_{i,j}$ generic pure state overlaps, using maximisation functions built in the \textit{Mathematica} language, specifically, NMaximize. We also tested if $(n-2)$-dimensional states could violate the inequalities $h_n(r)\leq 1$ by randomly generating sets of quantum states from the uniform, Haar measure. For $h_6(r)\leq 1$ we tested $10^{10}$ sets of 6 samples of ququarts and never violated the inequality. Both sampling sets of states and numerical tools finding local maxima for the functions $h_n(r)$ this property holds, i.e., no set of $n$ quantum states over an $(n-2)$-dimensional Hilbert space could violate the $h_n(r)\leq 1$ inequalities. 

\subsection{Probing the property of dimension witnessing with quadratic semidefinite programming}

{
Let us start by showing that maximizing the inequalities over pure states is sufficient, showing that our tests are robust, in the sense that the assumption of state purity is not needed for the dimension witnesses we obtained. Let us start by introducing some notation. We say that a given overlap tuple $r = (r_e)_e$, with respect to some graph $G$, has a quantum realization if there exists some Hilbert space $\mathcal{H}$ and some $\{\rho_i\}_i$ such that $r_e \equiv r_{ij} = \text{Tr}(\rho_i\rho_j)$. We denote a specific realization of $r$ as $r(\{\rho_i\}_i)$. Note that overlap tuples may have many realizations, such as $(1,1,1)$ in $G=C_3$, or none, such as $(1,0,1)$ in the same graph. For each realization $r(\{\rho_i\}_i)$, we denote $h_n(r(\{\rho_i\}_i))$ as the value that the functional $h_n$ has with respect to the overlap values $r$ with realization $r(\{\rho_i\}_i)$. Once more, a given value of $h_n(r)$ can be achieved with various quantum realizations of $r(\{\rho_i\}_i)$. 

With that, we can state the following result:

\begin{lemma}
        Let us consider graphs $K_n$, for some fixed value of $n \geq 3$. For any mixed state realization $r(\{\rho_i\}_i)$ from an ensemble $\{\rho_i\}_i$ there exists another pure state realization $r(\{\psi_i\}_i)$ from an ensemble $\{\psi_i\}_i$ such that $h_n(r(\{\rho_i\}_i)) \leq h_n(r(\{\psi_i\}_i))$.
\end{lemma}

    \begin{proof}
        Let any set of quantum states $\{\rho_i\}_i$ be associated with nodes on the event graph $G$. Each $\rho_i$ is a convex combination of pure states $\{\psi^{(i)}_{\omega_i}\}_{\omega_i \in \Omega_i}$ for some set $\Omega_i$ of pure states. Note that $h_n(r)$, for all fixed $n$, are by construction multilinear functionals satisfying,
    \begin{equation}\label{eq: multilinearity}
        \forall i,\rho_i = \sum_\omega \lambda_\omega^{(i)}\vert \psi_\omega^{(i)} \rangle \langle \psi_\omega^{(i)} \vert \implies
        h_n(r(\{\rho_i\}_i)) = \sum_{\omega_1,\dots,\omega_n} \lambda_{\omega_1}^{(1)}\dots \lambda_{\omega_n}^{(n)}h_n(r(\{\psi_{\omega_i}^{(i)}\}_i)). 
    \end{equation}
    To conclude the above, one needs to introduce some redundant values of $1 = \sum_{\omega_i}\lambda_{\omega_i}^{(i)}$. For concreteness, take $\{\rho_1,\psi_2,\psi_3\}$ and the inequality functional $h_3(r)$, 
    \begin{align*}
        &h_3(r(\{\rho_1,\psi_2,\psi_3\})) = \sum_{\omega_1}\lambda_{\omega_1}^{(1)} \vert \langle \psi_{\omega_1}^{(1)}| \psi_2 \rangle \vert^2+\sum_{\omega_1}\lambda_{\omega_1}^{(1)} \vert \langle \psi_{\omega_1}^{(1)}| \psi_3 \rangle \vert^2- \vert\langle \psi_2|\psi_3\rangle \vert^2\\
        &= \sum_{\omega_1}\lambda_{\omega_1}^{(1)}(\vert \langle \psi_{\omega_1}^{(1)}| \psi_2 \rangle \vert^2+ \vert \langle \psi_{\omega_1}^{(1)}| \psi_3 \rangle \vert^2-\vert\langle \psi_2|\psi_3\rangle \vert^2)\\
        &= \sum_{\omega_1}\lambda_{\omega_1}^{(1)} h_3(r(\{\psi_{\omega_1}^1,\psi_2,\psi_3\})),
    \end{align*}
    where we assumed that $\rho_1 = \sum_{\omega_1 \in \Omega_1}\lambda_{\omega_1}^{(1)}\vert \psi_{\omega_1}^{(1)} \rangle \langle \psi_{\omega_1}^{(1)} \vert $. One can then note that the same holds for all $n$ and for generic mixed ensembles $\{\rho_i\}_i$ when writing $h_n(r(\{\rho_i\}_i))$, implying Eq.~\eqref{eq: multilinearity}. One way to see this is simply iterating, for each element in the ensemble, the procedure exemplified by the $h_3$ case;
    \begin{align*}
        h_n(r(\{\rho_1,\rho_2,\dots,\rho_n\})) = \sum_{\omega_1}\lambda_{\omega_1}^{(1)}h_n(r(\{\psi_{\omega_1}^{(1)}, \rho_2, \dots, \rho_n\})) = \sum_{\omega_1,\omega_2}\lambda_{\omega_1}^{(1)}\lambda_{\omega_2}^{(2)}h_n(r(\{\psi_{\omega_1}^{(1)}, \psi_{\omega_2}^{(2)}, \dots, \rho_n\})) = \dots \\= \sum_{\omega_1,\dots,\omega_n}\lambda_{\omega_1}^{(1)}\dots \lambda_{\omega_n}^{(1)}h_n(r(\{\psi_{\omega_1}^{(1)},\psi_{\omega_2}^{(2)},\dots,\psi_{\omega_n}^{(n)}\})).
    \end{align*}
    We can collectively write $s = (\omega_1,\dots,\omega_m)$ and define $q_s = \lambda_{\omega_1}^{(1)}\dots \lambda_{\omega_m}^{(m)}$. Because each set of weights $\{\lambda_{\omega_i}^{(i)}\}_{\omega_i \in \Omega_i}$ corresponds to convex weights, i.e., $\sum_{\omega_i}\lambda_{\omega_i}^{(i)} =1 $ with $0\leq \lambda_{\omega_i} \leq 1$ we get that $\{q_s\}_s$ is also a set of convex weights. With this simplified notation we have that Eq.~\eqref{eq: multilinearity} becomes 
    $$h_n(r(\{\rho_i\}_i)) = \sum_sq_s h_n(r(\{\psi^{(i)}_s\}_i))$$ with $\sum_sq_s=1$ and $0\leq q_s\leq 1$. In words, the linear-functional $h_n$ realized by overlaps between general quantum states can be written as the convex combination of the same functional realized by overlaps between pure states. Choosing now a particular ensemble $s^\star$ such that $\forall s, h_n(r(\{\psi_{s^\star}^{(i)}\}_i)) \geq h_n(r(\{\psi_s^{(i)}\}_i))$ we see that $$h_n(r(\{\rho_i\}_i)) = \sum_sq_s h_n(r(\{\psi^{(i)}_s\}_i))\leq \sum_sq_s h_n(r(\{\psi_{s^\star}^{(i)}\}_i)).$$
    Since $\sum_sq_s=1$ we have that $h_n(r(\{\rho_i\}_i)) \leq h_n(r(\{\psi_{s^\star}^{(i)}\}_i))$.
    \end{proof}
}

{Our goal now is to} push the numerical results of Table \ref{tab:kn} for instances of $K_n$ for larger $n$ using techniques of semidefinite programming in order to see if there is the possibility of probing the dimension witness nature of the family of inequalities for larger complete graphs. The major drawback with respect to the maximization over parameters, as we will see, is that we cannot \textit{read} the states that realize the given violation from the final result. In order to show that $h_n$ are dimension witnesses for higher values of $n$, we show how to re-write the expression defining this family of inequalities into an SDP. We start with two simple lemmas, the first taken from Ref.~\cite{brunner2013dimension}. 

\begin{lemma}\label{lemma: hn+=trX2}
Let $h_n^+(r)$  be the \textit{sum} of all overlaps for a graph $K_n$. Assume also that the quantum realization of the graph is given by the set of states $\{\vert \psi_i \rangle\}_{i=1}^n$.    Then, we have that
\begin{equation}
    h_n^+(r) = \frac{n^2}{2}\mathrm{Tr}(X^2)- \frac{n}{2}
\end{equation}
where $X = \frac{1}{n}\sum_{i=1}^n \vert \psi_i \rangle \langle \psi_i \vert$. 
\end{lemma}

\begin{proof}
    Noting that $X^2$ is given by, 
\begin{equation*}
    X^2 = \frac{1}{n^2} \sum_{i,j=1}^n \vert \psi_i \rangle \langle \psi_i \vert \psi_j \rangle \langle \psi_j \vert  = \frac{1}{n^2}\sum_{i\neq j}\vert \psi_i \rangle \langle \psi_i \vert \psi_j \rangle \langle \psi_j \vert  + \frac{1}{n^2}\sum_{i=j}\vert \psi_i \rangle \langle \psi_i \vert 
\end{equation*}
\begin{equation*}
    \text{Tr}(X^2) = \frac{1}{n^2}\sum_{i\neq j}r_{i,j}+ \frac{1}{n^2}\sum_{i=1}^n \vert \langle \psi_i \vert \psi_i \rangle \vert^2 = \frac{1}{n^2}\sum_{i\neq j}r_{i,j} + \frac{1}{n} = \frac{2}{n^2}\sum_{i<j}r_{i,j}+\frac{1}{n}
\end{equation*}
where we have used the fact that $r_{i,j} = r_{j,i}$. Now, note that in fact $h_n^+(r) = \sum_{i < j}r_{i,j}$. From this we have our relation by inverting the last equation,
\begin{equation*}
    \text{Tr}(X^2) = \frac{2}{n^2}h_n^+(r) + \frac{1}{n} \implies h_n^+(r) = \frac{n^2}{2}\left(\text{Tr}(X^2) - \frac{1}{n}\right) = \frac{n^2}{2}\text{Tr}(X^2)-\frac{n}{2}
\end{equation*}
as we wanted to show.
\end{proof}

We may also note that, due to the recurrence description of $h_n(r)$ it is possible to use only $h_n^+(r)$ to describe the inequality.

\begin{lemma}\label{lemma: hn = hn+-2hn+}
 The inequality functional $h_n(r)$ is equal to $h_n(r) = h_n^+(r)-2h_{n-1}^+(r)$.
\end{lemma}

\begin{proof}
As we have that 
\begin{equation*}
    \sum_{i=1}^n r_{i,n} = h_n^+(r)-h_{n-1}^+(r)
\end{equation*}
We can see that,
\begin{equation*}
    h_n(r) = \sum_{i=1}^nr_{i,n} - h_{n-1}^+(r) = h_n^+(r)-h_{n-1}^+(r)- h_{n-1}^+(r)
\end{equation*}
as we wanted to show. 
\end{proof}

For the sake of clarity, let's consider an example of the above description for the lemma. Recalling that $h_3^+(r) = r_{1,2}+r_{1,3}+r_{2,3}$ then, 
\begin{align*}
    h_4^+(r)-2h_3^+(r) &= r_{1,2}+r_{1,3}+r_{1,4}+r_{2,3}+r_{2,4}+r_{3,4}-2(r_{1,2}+r_{1,3}+r_{2,3})\\
    &=r_{1,4}+r_{2,4}+r_{3,4}-r_{1,2}-r_{1,3}-r_{2,3} = h_4(r)
\end{align*}

This is equivalent to $h_4(r)$ up to the relabel $4\to 1$ and $1 \to 4$. We can use the lemmas above to write $h_n(r)$ only in terms of $X$, which eventually lead a semidefinite programming description of the problem of maximizing $h_n(r)$ for finding quantum violations.

\begin{theorem}
    The quantum realizations of the inequality functional $h_n(r)$ can be expressed as a quadratic semidefinite program (SDP) optimized over  quantum states $X$. The resulting optimized values provide upper bounds for the values of $h_n(r)$ that can be reached with quantum theory for { any $d$.}. 
\end{theorem}

\begin{proof}
    From lemma \ref{lemma: hn = hn+-2hn+} and lemma \ref{lemma: hn+=trX2} together we can write $h_n(r)$ as,
    \begin{equation}
        h_n(r) = \left(\frac{n^2}{2}\text{Tr}(X^2)-\frac{n}{2}\right) - 2\left(\frac{(n-1)^2}{2}\text{Tr}(X_\star^2)-\frac{n-1}{2}\right)
    \end{equation}
where we have that $X_\star = \frac{1}{n-1}\sum_{i=1}^{n-1}\vert \psi_i \rangle \langle \psi_i \vert$. In other words, the second term of $h_n$ lacks $\vert \psi_n\rangle \langle \psi_n \vert$. It is clear that we may write $X$ in terms of $X_\star$ provided that we fix $\vert \psi_n \rangle$. As $h_n(r)$ is a projective-unitarily invariant functional of a set of states we can unitarily transform any set of states such that  $\vert \psi_n\rangle = \vert 0 \rangle \in \mathcal{H} \simeq \mathbb{C}^d$ is the reference $d$-dimensional canonical basis vector of $\mathbb{C}^d$. In this case, we have that
\begin{equation*}
    X = \frac{1}{n}\sum_{i=1}^n\vert \psi_i \rangle \langle \psi_i \vert = \frac{1}{n}\left(\sum_{i=1}^{n-1}\vert \psi_i\rangle\langle\psi_i \vert +\vert \psi_n \rangle \langle \psi_n \vert \right) = \frac{1}{n}((n-1)X_\star+\vert 0\rangle\langle 0 \vert )
\end{equation*}
which implies that,
\begin{equation*}
    X^2 = \frac{1}{n^2}((n-1)X_\star+\vert 0\rangle\langle 0 \vert )((n-1)X_\star+\vert 0\rangle\langle 0 \vert ) = \frac{1}{n^2}\left( (n-1)^2X_\star^2 + (n-1)X_\star \vert 0\rangle\langle 0\vert + (n-1)\vert 0\rangle\langle 0 \vert X_\star + \vert 0\rangle\langle 0\vert   \right)
\end{equation*}
from which we infer that,
\begin{equation*}
    \text{Tr}(X^2) = \frac{1}{n^2}\left((n-1)^2\text{Tr}(X_\star^2) + 2(n-1)\langle 0 \vert X_\star \vert 0\rangle + 1\right).
\end{equation*}
This last expression allows us to write $h_n(r)$ in terms of $X_\star$ and $\vert 0\rangle \langle 0 \vert$ only.

\begin{align*}
    h_n(r) &= \frac{n^2}{2}\left( \frac{1}{n^2}\left((n-1)^2\text{Tr}(X_\star^2) + 2(n-1)\langle 0 \vert X_\star \vert 0\rangle + 1\right) \right) - \frac{n}{2}-2\left(\frac{(n-1)^2}{2}\text{Tr}(X_\star^2)-\frac{n-1}{2}\right)\\
    &=\frac{(n-1)^2}{2}\text{Tr}(X_\star^2)+(n-1)\text{Tr}(\vert 0\rangle\langle 0\vert X_\star)+\frac{1}{2}-\frac{n}{2}-(n-1)^2\text{Tr}(X_\star^2)+n-1\\
    &=-\frac{(n-1)^2}{2}\text{Tr}(X_\star^2)+(n-1)\text{Tr}(\vert 0\rangle\langle 0\vert X_\star)+\frac{n-1}{2}
\end{align*}

If we fix $C := \vert 0\rangle\langle 0 \vert$ for $\vert 0\rangle \in \mathbb{C}^d$ the above form of $h_n$ is quadratic in $X_\star$. We can then write the following quadratic SDP problem,
\begin{align}
    \max_{X_\star \in \mathbb{C}_{d,d}} & A_n \langle X_\star,X_\star \rangle + B_n \langle C,X_\star \rangle + C_n \\
    \text{subject to } & X_\star \geq 0\\
    &\text{Tr}(X_\star)=1
\end{align}
for any $2 \leq d \leq n-1 $ as $X_\star$ lives in the span of $\{\vert \psi_i \rangle\langle \psi_i \vert \}_{i=1}^n$. Above we simply set, $A_n := -(n-1)^2/2$, $B_n := (n-1)$ and $C_n=(n-1)/2$. We also use the common description of the trace inner product $\text{Tr}(A^\dagger B) = \langle A,B\rangle$.  
\end{proof}

This theorem immediately shows that we may find upper bounds, for any dimension {$d$}, on the maximum values of $h_n(r)$ from a set of states $\{\vert \psi_i \rangle\}_{i=1}^n$, as the set of all matrices of the form $X_\star = \frac{1}{n-1}\sum_{i=1}^{n-1}\vert \psi_i \rangle \langle \psi_i \vert $ is only a subset of all possible density matrices. The theorem also shows that it is unnecessary to search for dimensions $d>n-1$ as the problem is defined for matrices $X_\star$ that are inside the span of $\{\vert \psi_i \rangle \}_{i=1}^{n-1}$, which is at most $(n-1)$ dimensional. 

The interesting aspect of transposing the problem into an SDP is that the problem becomes computationally efficient, {polynomial in the memory and computational resources~\cite{tavakoli2023semidefinite,li2018quadraticSDP}}. In figure \ref{fig: sdp} we implement the above problem and find solutions for $n$ up to $800$. Such a gain allows us to study upper bounds of quantum violations for all integers  $4 \leq n \leq 800$. The main idea for these simulations is to provide strong numerical evidence that the family $K_n$ is a dimension witness for all possible values of $n\geq 4$. As {this} might be interesting for when the Hilbert space dimension grows exponentially, we have also considered the case of dimension $d=2^q$ where $q$ represents number of qubits. We show numerically that $2^q-2$ dimensional states cannot violate inequalities $h_{2^q}(r)\leq 1$, up to $q=12$. The results are shown in Fig.~\ref{fig: sdp qubits}. { As our quadratic SDP is well behaved, we have used interior point methods, using solvers available in $\mathrm{CVXPY}$. 

We used the Splitting Cone Solver (SCS) and
Python Software for Convex Optimization (CVXOPT), and both converged to the same values up to numerically instabilities. Results presented are those for SCS only. We have used these solvers since they are open-source. While not devoted to quadratic SDP optimization, they have converged fast for all points considered. As these are convex optimization tools applied to a quadratic optimization problem, we are not guaranteed that the results are tight; still we managed to find (local) optimal values that agreed with other methods (\textrm{NMaximize} from \textit{Mathematica}) and our experimental implementations, providing further evidence that our dimension witnesses hold significantly beyond the regimes we have tested experimentally. The SDP code may be found in Ref.~\cite{giordani2023code}.}

\begin{figure}[t]
    \centering
    \includegraphics[width=\textwidth]{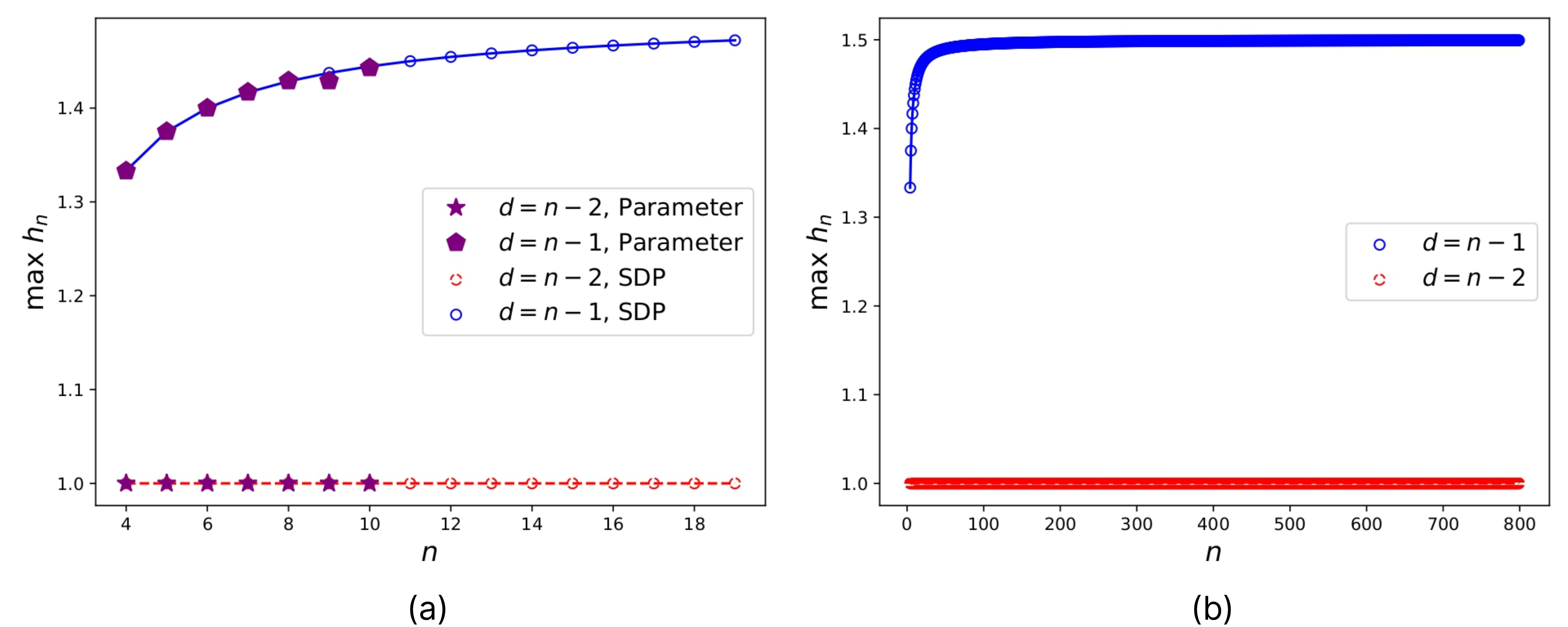}
    \caption{\textbf{Numerically testing dimension and coherence witness using semidefinite programming.} (Color online) We use a quadratic semidefinite programm (SDP) to generate (tight) upper bounds on the maximum value of $h_n(r)$ for any value $n\geq 4$, for any dimension $d\leq n-1$. Plot (a) shows the result for $n$ up to 19. We compare the results obtained using SDP with the results obtained by maximizing $h_n$ over all parameters of all quantum states present in a given complete graph $K_n$. Purple stars and pentagons corresponds to points from table \ref{tab:kn}, and we use those to benchmark the SDP results. Cases $n=4,5,6$ were also investigated experimentally. We see that for all those values of $n$ the inequality $h_n(r)\leq 1$ is, at the same time, bounding coherence and dimension. Plot (b) we use the fact that the quadratic SDP can be used to estimate the validity of the conjecture that only states with dimension $d\geq n-1$ can violate the inequality $h_n(r)\leq 1$ for $4\leq n\leq 800$. We see that the maximum quantum violation as $n\to \infty$ is given by $1/2$ and that the upper bound in the value for $d=n-2$ is always 1. In both plots (a) and (b) open circles correspond to maximum values of $h_n(r)$ found using the SDP. Blue full lines mark violations of the $K_n$ inequality, and have $d=n-1$ while red dashed lines correspond to points that do not violate the inequality and have $d=n-2$.}
    \label{fig: sdp}
\end{figure}

\begin{figure}[t]
    \centering
    \includegraphics[width=\textwidth]{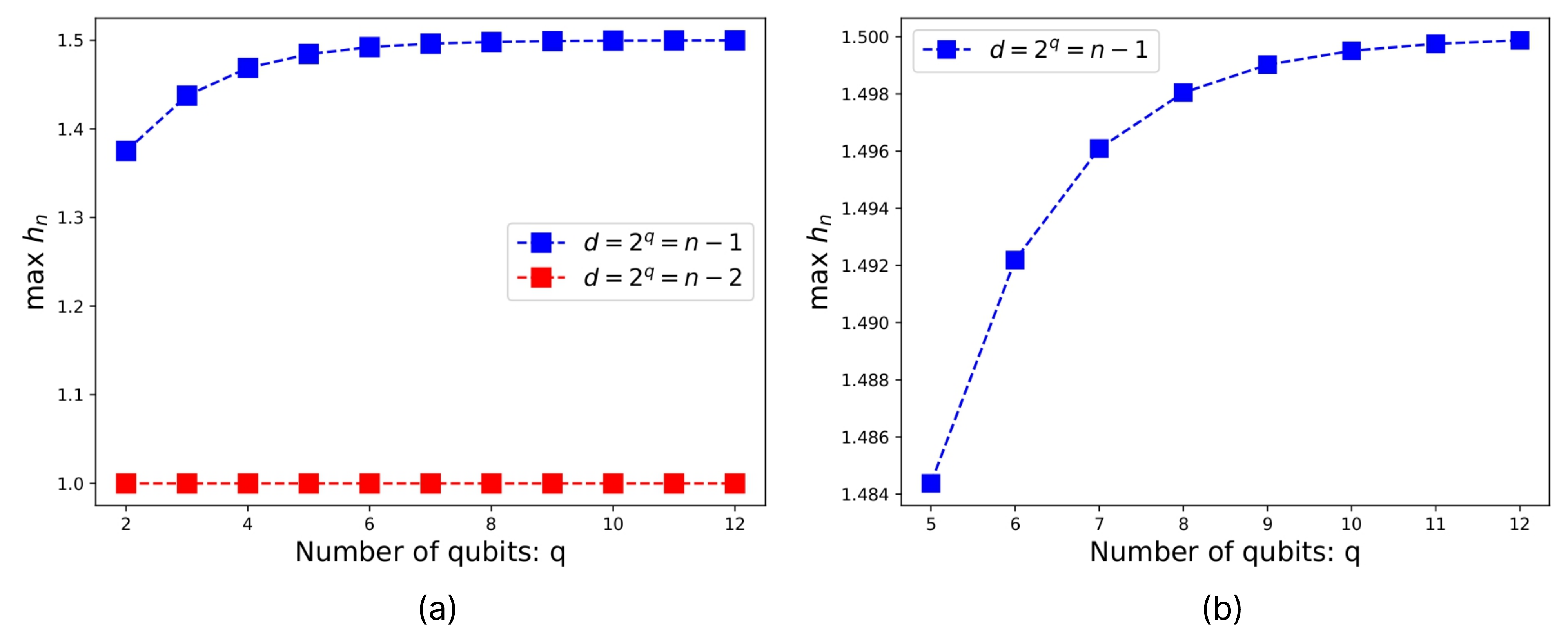}
    \caption{\textbf{Numerically testing dimension and coherence witness using semidefinite programming for qubits}. We consider the same simulations from Fig.~\ref{fig: sdp} but using qubits, and therefore going to exponentially large Hilbert space dimensions. We see that $h_n$ remains a dimension witness for $d$ up to $2^{12}$. In (a) we see that there is a saturation of $1.5$, and in (b) we highlight the changes with a higher precision.}
    \label{fig: sdp qubits}
\end{figure}

{
\subsection{Assumptions for the task of dimension witnessing with overlap inequalities}

Certification tasks can depend on the underlying assumptions made on the data generated by quantum devices, and on some promises on the experimental implementation or on the device itself. For a review on different certification protocols, and the most well-known assumptions we refer the reader to Ref.~\cite{eisert2020certification}. 

For the certification of coherent qudits we rely on the violation of overlap-based inequalities. The assumptions made are as follows:

\begin{enumerate}
    \item We assume that the data corresponds to  two-state overlaps. In other words, we assume that we are given values $r_{ij} = \text{Tr}(\rho_i\rho_j)$, for some $\rho_i,\rho_j$ with respect to some $\mathcal{H}$. 
    \item We assume that states are prepared independently and identically distributed, in order to generate the two-state overlap data. 
\end{enumerate}

These are the only assumptions made; drawing a parallel with a black-box certification task, we use only the statistics generated by the device -- under a promise over the possible measurements, i.e., those generating two-state overlap data -- to witness a lower bound on the dimensionality of the Hilbert space that has generated the data. We \textit{do not} assume: (i) upper bound in Hilbert space dimension, (ii) restriction on possible states being prepared, (iii) a specific system structure (multiqubits, multiqudits), (iv) a specific hardware implementation, (v) purity of the states. Our test is not device-independent. We also do not make the claim that our test is semi-device independent, leaving such a discussion for future work. 

From the first point above, it is assumed implicitly that the data satisfy the symmetry $r_{ij} = r_{ji}$. While this is certainly not going to hold for a generic prepare-and-measure scenario, it is possible to experimentally bound the degree of deviation of such a property, to estimate how far one is to the overlap assumption. In Supplementary Note 7 we have bounded such deviations in the estimation of $h_5(r)$. 

}
\section{Brief overview of generalized contextuality}
\label{appendix: overview_cont}

In this Supplementary Note we briefly review generalized contextuality, and provide a robust version for the proof of contextual advantage from Ref.~\cite{wagner2022coherence}. The notion of generalized contextuality has been connected with various quantum advantage results~\cite{spekkens2009parity,saha2019preparation,schmid2018contextual,ambainis2019parity,chailloux2016optimal,lostaglio2020certifying,lostaglio2020contextual,flatt2022contextual,rochicarceller2022quantumvs}. See Refs.~\cite{budroni2021kochen,mazurek2016experimental,mazurek2021experimentally,Zhan2017experimental,Zhang2022experimental} for known experimental tests that have probed contextuality. Noncontextuality is a constraint over probabilistic data, considered as a notion of classicality because it is connected either with classical Hamiltonian mechanics~\cite{lostaglio2020certifying}, or with models that have undergone decoherence or that verify the emergence of objectivity~\cite{baldi2021emergence}. Whenever said here, we consider \textit{advantage} to account for the difference between success rates in information tasks achieved with contextuality, as opposed to what can be achieved using noncontextual models. As mentioned in the main text, noncontextual models can be  intuitively viewed as a comprehensive class of   classical explanations, similar to local hidden variable models for Bell nonlocality, that permit classical interpretations while, e.g., explaining data generated by 
 some  entangled states~\cite{werner1989quantum}. Noncontextual models correspond to the generalization of local models for single systems, and there are known ways to map correlations between the two~\cite{wright2022invertible}. In the next sections we will also use recently introduced numerical tools~\cite{selby2021accessible,selby2022opensource,rossi2022contextuality} that simplify tests of generalized contextuality.

Starting with a description of generalized contextuality as introduced in Ref.~\cite{spekkens2005contextuality}, we need two concepts: operational theories and ontological models, and we start by  presenting the first. A generic way of describing laboratory operations is by prescribing finite lists of instructions. \textit{Preparation procedures} $P\in\mathcal{P}$, e.g.,  might represent instructions as ``generate a single vertically polarized photon''. Similarly, outcomes $k\in K$ of measurement procedures $M\in \mathcal{M}$ (we call the pair $k \vert M$ a \textit{measurement event}) also represent instructions, e.g., ``upon passing a photon through a polarizing {beam splitter} the detector corresponding to horizontal polarization clicks''.  The collection of all such operations, together with some \textit{composition rules}, and \textit{convex combinations of procedures}, define an operational theory. If we also associate the composition of a preparation $P$ followed by measurements $M$ with outcomes $k$ to some probabilistic interpretation rule $p$, for which we read $p(k \vert M, P)$ we have an operational-probabilistic theory. 

Quantum theory can be viewed as one possible theory in this landscape, where $p$ is given by the Born rule, operations $P$ represent the preparation of quantum states $\rho$ and measurement events $k \vert M$ correspond to positive operator-valued measures (POVM) elements $E_k$.

Two different procedures $P_1,P_2 \in \mathcal{P}$ can be indistinguishable with respect to $p$, in the sense that for every conceivable measurement effect $k \vert M $ the two generate equal statistics $p(k \vert M, P_1) = p(k \vert M, P_2)$. When this happens, we say that $P_1$ and $P_2$ are operationally equivalent and write $P_1 \simeq P_2$. Similarly, two effects $k_1\vert M_1, k_2 \vert M_2$ are said to be operationally equivalent when for any conceivable $P \in \mathcal{P}$ one has $p(k_1 \vert M_1 , P) = p(k_2 \vert M_2, P)$. Note that, if two procedures, for instance $P_1,P_2 \in \mathcal{P}$  are equivalent, the distinction described by the labels $1$ and $2$ merely capture the context in which they were performed~\cite{spekkens2005contextuality}. 

Generalized noncontextuality is a restriction over ontological models. In such models~\cite{harrigan2010einstein}, one assumes that physical systems can be completely characterized by sets of variables $\lambda$ pertaining to a larger space of such elements $\lambda \in \Lambda$. The model prescribes a causal account of the operational-probabilistic theory; to any preparation $P \in \mathcal{P}$ the model explains that, with probability $\mu(\lambda \vert P)$ over $\Lambda$,  some variable $\lambda$ was prepared via procedure $P$. Similarly, for any measurement effect $k\vert M$, the model explains outcomes $k$, resulting from some $\lambda \in\Lambda$ and $M$ through functions $\xi(k \vert M, \lambda)$. The model then connects theory with experimental data using the rule

\begin{equation}
    p(k \vert M, P) = \sum_{\lambda \in \Lambda} \mu(\lambda \vert P) \xi(k \vert M, \lambda).
\end{equation}

The restriction that noncontextuality~\cite{spekkens2005contextuality} imposes over such models is that operationally equivalent procedures $P_1 \simeq P_2$ must be represented equally by the model: $\mu(\lambda \vert P_1) = \mu(\lambda \vert P_2), \forall \lambda \in \Lambda$. Similarly for the measurement effects, $k_1 \vert M_1 \simeq k_2 \vert M_2$ implies $\xi(k_1 \vert M_1 , \lambda) = \xi(k_2 \vert M_2 , \lambda ), \forall \lambda \in \Lambda$. Intuitively, the model \textit{explains} that one cannot distinguish operationally equivalent procedures $P_1 \simeq P_2$ simply because they correspond to the exact same ontological counterparts $\mu(\cdot | P_1) = \mu(\cdot | P_2)$.

In order to test generalized contextuality, we first consider the approach that characterizes prepare-and-measure scenarios described by fragments of finite sets of operational elements in a theory. We shall refer to it as the algebraic approach~\cite{schmid2023addressing}, or equivalently the inequalities approach. For a given finite set of operational elements, it is possible (but generally hard) to find the complete set of generalized noncontextuality inequalities bounding noncontextual explanations for such scenarios~\cite{schmid2018all}. Using semidefinite program (SDP) techniques it is also possible to bound the set of quantum correlations~\cite{tavakoli2021bounding,chaturvedi2021characterising}. 

As was shown in Refs.~\cite{wagner2022coherence,Wagner2022}, the inequalities associated with graph $K_3$, constituting the  class $r_{1,2}+r_{1,3}-r_{2,3}\leq 1$, can be mapped into robust generalized noncontextuality inequalities of  specific prepare-and-measure scenarios. Formally, the prepare and measure scenario, in its robust form, is described as one in which there are $6$ preparation procedures satisfying 
\begin{equation}\label{eq: op equivalences}
    \frac{1}{2}P_1 + \frac{1}{2}P_{1^\perp} \simeq \frac{1}{2}P_2 + \frac{1}{2}P_{2^\perp} \simeq \frac{1}{2}P_3 + \frac{1}{2}P_{3^\perp}
\end{equation}
and three other measurement procedures, for which we assume no operational equivalences, but that satisfy

\begin{equation}
p(M_i\vert P_i)\geq 1-\varepsilon_i, p(M_i\vert P_{i^\perp}) \leq \varepsilon_i, \forall i \in \{1,2,3\}, \,\varepsilon_i \geq 0.    
\end{equation}
Above we use the notation $p(0\vert M_i , P_j) \equiv p(M_i \vert P_j)$. For our experiment, $p(M_i \vert P_j) = r_{i,j}$. For such a scenario it is possible to map the $K_3$ inequality into a robust noncontextuality inequality of the form 
\begin{equation*}
    -p(M_1\vert P_2)+p(M_1 \vert P_3)+p(M_2\vert P_3) \leq 1 + \varepsilon_1 + \varepsilon_2+\varepsilon_3.
\end{equation*}

The quantities $\varepsilon_i$ theoretically capture the fact that, in a real experiment, the  measurements $M_i$ do not perfectly discriminate between $P_i$ and $P_{i^\perp}$. Such imperfections might occur operationally in case $M_i$, $P_i$ or $P_{i^\perp}$ do not correspond to the ideal intended procedures. This last inequality allows us to analyse the claim of contextual advantage for quantum interrogation, from the algebraic perspective. We will do this in the following, showing that our scenario satisfies the operational constraints considered \textit{and} violates the noncontextuality inequality above.

\section{Robust account of contextuality in quantum interrogation}\label{appendix: robust_inter}

Let us start by analysing the robustness of the effects of contextuality in the interrogation task made in Ref.~\cite{wagner2022coherence} due to depolarizing noise. Quantum interrogation assumes the possibility of preparing six different quantum states $\{\vert 0\rangle, \vert 1\rangle, \vert \theta \rangle, \vert \theta^\perp \rangle, \vert -\theta \rangle, \vert -\theta^\perp \rangle\}$, $\theta \in [0,\pi]$, and measuring over all of these states as well, where the states relevant for the quantum interrogation correspond to those generated by the {beam splitter},
\begin{align*}
    \vert \theta \rangle = \cos(\theta)\vert 0\rangle + \sin(\theta)\vert 1\rangle.
\end{align*}

States $\vert \theta \rangle$ result from unitary operations implemented by the first {beam splitter} in Fig.~3(a) in the main text, when {single photons} are input in mode 0. The states $\vert - \theta \rangle$ correspond to the rotation that perfectly destroys the interference generated in the first {beam splitter}, described as unitary rotations generated by the second {beam splitter} in Fig.~3(a), while the third {beam splitter} in this figure represents the bomb in the interrogation scheme. The {programmability} of the device allows for such a freedom in the preparation of those states. Antipodal states can be equally prepared by inputting a photon in the second arm of the MZI. The {programmability} of the device allows that each of those states can be prepared and all the measurements $M_0 := \{\vert 0 \rangle \langle 0 \vert , \vert 1 \rangle \langle 1 \vert \}$, $M_\theta := \{\vert \theta \rangle \langle \theta \vert, \vert \theta^\perp \rangle \langle \theta^\perp \vert ,\}$ can also be performed by programming the {beam splitters} accordingly and gathering statistics from counts in the {APDs}. 

The trade-off between what can be achieved with quantum theory and any noncontextual model is described by the noncontextuality inequality
\begin{equation}\label{eq: noncontextuality inequality}
    p(0\vert M_0,P_\theta) + p(1 \vert M_0, P_{-\theta})-p(\theta \vert M_{\theta},P_{-\theta}) \leq 1 + \varepsilon_0+\varepsilon_\theta+\varepsilon_{-\theta},
\end{equation} related to the prepare-and-measure scenario described before. Noise affects the violations of the inequalities in two different ways: first, it \textit{increases} the noncontextual bound on the right, and second, it decreases the values of the statistics. These two factors suggest that noise is very detrimental to proofs of contextuality that consider specific (target) equivalences. This is common in proofs of contextual advantage, c.f Refs.~\cite[Fig. 3, pg. 7]{lostaglio2020contextual}, and \cite[Fig. 7, pg. 11]{schmid2018contextual}.

Recalling that, operationally, the efficiency of the task is characterized by the quantity 
\begin{equation*}
    \eta = \frac{p(1\vert M_0 , P_{-\theta})p(0 \vert M_0 , P_\theta)}{p(1\vert M_0 , P_{-\theta})p(0 \vert M_0 , P_\theta)+p(1\vert M_0, P_\theta)}
\end{equation*}
where $p(0 \vert M_0, P_{\theta})$ corresponds to the probability that, upon the measurement $M_0$ ('which-way measurement' deciding the path that the {single photon} took) the photon is found in the mode corresponding to state $\vert 0\rangle$, similarly $p(1\vert M_0, P_\theta)$ to the one in which the photon is found in the complementary mode (and is interpreted as the 'bomb' exploding in the Elitzur-Vaidman experiment) while $p(1\vert M_0, P_{-\theta})$ corresponds to the probability that a which way measurement finds a photon in mode $\vert 1\rangle$ upon the preparation $P_{-\theta}$. Note that imposing the symmetry $p(1\vert M_0, P_{-\theta})=p(1\vert M_0, P_{\theta})$ simplifies the equation into
\begin{equation*}
    \eta = \frac{p(0\vert M_0, P_\theta)}{p(0\vert M_0, P_\theta) + 1}.
\end{equation*}

The strategy from now on will be to use the theoretical description of the quantum experiment in terms of noisy quantum states affected by the channel $\mathcal{D}_\nu$. We therefore define, for each state $\vert \psi_i \rangle \langle \psi_i \vert$, a state described by the action of a depolarizing channel $\mathcal{D}_\nu$
\begin{equation}\label{eq: depo channel}
    \mathcal{D}_\nu(X) := (1-\nu)X + \nu \frac{I_d}{d}\text{Tr}(X),
\end{equation}
with $I_d$ a $d\times d$ identity matrix, implying that $I_2/2$ corresponds to the maximally mixed qubit state. The factor $\nu$ represents the amount of noise.  We map  measurement effects $0\vert M_0 \mapsto \mathcal{D}_\nu(\vert 0 \rangle \langle 0 \vert), 1 \vert M_0 \mapsto \mathcal{D}_\nu(\vert 1\rangle \langle 1\vert )$ and similarly for all $\theta$,  $0|M_{\theta} \mapsto \mathcal{D}_\nu(\vert \theta \rangle \langle \theta \vert)$\footnote{Technically, each measurement effect and each quantum state correspond to an equivalence class of all possible procedures that prepare/measure them. We abuse notation here, by simply saying that we map operational measurement effects into specific instances of quantum effects.}. The same  for the preparations. Therefore, we have that the optimal quantum strategy obtained under the effect of a channel $\mathcal{D}_\nu$ now reads

\begin{equation}\label{eq: robust efficiency}
    \eta = \frac{\text{Tr}(\mathcal{D}_\nu(\vert 0 \rangle \langle 0 \vert ) \mathcal{D}_\nu(\vert\theta \rangle \langle \theta \vert))}{\text{Tr}(\mathcal{D}_\nu(\vert 0 \rangle \langle 0 \vert ) \mathcal{D}_\nu(\vert\theta \rangle \langle \theta \vert)) + 1} = \frac{\text{Tr}(\rho_0\rho_\theta)}{\text{Tr}(\rho_0\rho_\theta)+1},
\end{equation}
where for simplicity we write $\rho_s \equiv \mathcal{D}_\nu(\vert s \rangle \langle s \vert) $. For the noncontextuality bound, we can see that $\varepsilon_0$,  $\varepsilon_\theta$ and $\varepsilon_{-\theta}$ will be functions of $\nu$,
\begin{equation}
    1-\varepsilon_0 = p(0\vert M_0 , P_0) = \text{Tr}(\rho_0\rho_0) = 1+\frac{\nu^2}{2}-\nu
\end{equation}
and it is easy to see that $\varepsilon_0 = \varepsilon_\theta = \varepsilon_{-\theta}$. Therefore, $\varepsilon_i = \nu - \frac{\nu^2}{2}$ for $i \in \{0,\theta,-\theta\}$. In such a way, we have that the robust noncontextual bound provided by the prepare-and-measure noncontextuality inequality will be 
\begin{equation}\label{eq: noncontextual efficiency}
    \eta^{NC}(\theta,\nu) \equiv \frac{1 + \text{Tr}(\rho_\theta\rho_{-\theta}) - \text{Tr}(\rho_1\rho_{-\theta}) + 3\left(\nu - \frac{\nu^2}{2}\right)}{\text{Tr}(\rho_0 \rho_\theta) + 1}.
\end{equation}

Whenever $\eta > \eta^{NC}$, we observe an advantage provided by quantum contextuality in the interrogation task. The function $\eta^{NC}(\theta,\nu)$ robustly characterizes the validity of the no-go result, and not necessarily the actual existence of a noncontextual model that reproduces the data.  In Fig.~\ref{fig:robust}(a) we plot both the efficiency that can be achieved with quantum theory (blue) $\eta$ and the noncontextual bound (pink) $\eta^{NC}$ as a function of the parameter $\theta$ that characterizes the transmissivity of the {beam splitters} in the MZI, and the amount of noise $\nu$ captured by the channel from Eq.~\eqref{eq: depo channel}. The curve in which the two meet characterizes the degree of noise $\nu$ for which no advantage can be claimed. In Fig.\ref{fig:robust}(b) we plot the robustness to the noise parameter $\nu$ as a function of $\theta$ as described by the operational prepare-and-measure scenario considered. In this case, for $\nu>0.057$ we would lose the gains guaranteed by contextuality in the protocol.

\begin{figure}[ht]
    \centering
    \includegraphics[width=1\textwidth
]{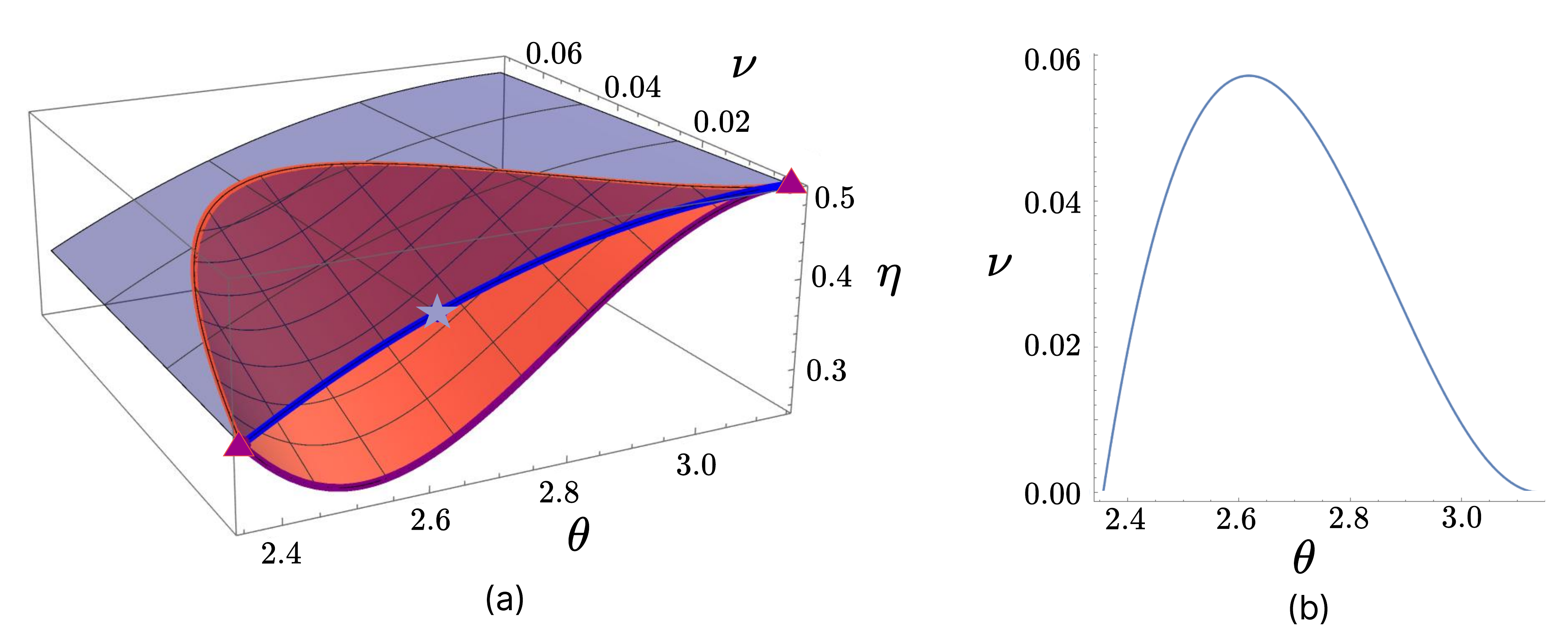}
    \caption{\textbf{Efficiency of quantum interrogation.} (a) Efficiency $\eta$ operationally described in terms of successfully detecting the object as a function of a single parameter $\theta$ encoding the difference in transmission/reflectivity rates in the {beam splitters}. In the upper curve (blue) we have what is achieved with quantum theory and in the lower curve (pink/purple) the upper bound on the efficiency of noncontextual models. Robustness is captured by the degree of depolarizing noise captured by a parameter $\nu$ from the channel $(1-\nu)\rho +\nu I/2$, for any state $\rho$ prepared inside the interferometer. The purple triangles corresponds to the efficiencies $\eta = 1/3$ and $\eta=1/2$ achieved using the toy model from Ref.~\cite{catani2021whyinterference}. The blue star corresponds to a choice of $\theta=5\pi/6$ that can be translated into a proof of contextuality as originally presented in Ref.~\cite{spekkens2005contextuality}, and that achieves $\eta(5\pi/6,0)=0.428$. If $\nu>0.057$ the gap is lost and it is impossible to conclude contextual advantage for the task. (b) Curve marking, for each value of $\theta$, the noise level for which any claim of contextual advantage is lost. Visually, this curve marks the intersection of the two curved regions present in (a).}
    \label{fig:robust}
\end{figure}

From the analysis in Fig.~\ref{fig:robust}, it is clear that for the operational prepare-and-measure (PM) scenario considered, the efficiency of quantum interrogation is fairly sensitive to depolarising noise. Each fixed value of $\theta$ characterizes a specific PM scenario. We can analyse if our results corroborate the hypothesis that the experimentally observed efficiency is higher than those achieved by noncontextual models, considering noise in the device. In order to do so, we study contextuality in the device for the case of values $r=0.75$, where $\theta = 5\pi/6$. In what follows, we will first experimentally test if the operational equivalences that the MZI states should satisfy are reproduced by the device within experimental error, and then use these states to calculate the value of the noncontextuality inequality Eq.~\eqref{eq: noncontextuality inequality}. A violation will robustly indicate that the device is witnessing generalized contextuality within the so-called algebraic/inequalities approach. We will later use a different approach, the so called geometric/general probabilistic theories (GPT) approach~\cite{schmid2021characterization,shahandeh2021contextuality}, to obtain improved results for robustness to depolarizing noise, and obtain experimental evidence of contextuality also from the GPT perspective. We finish the discussion by showing that the efficiency experimentally achieved in the quantum interrogation task \textit{cannot} be achieved with noncontextual models, and comment on open loopholes for our test.

\section{Loophole analysis and benchmarking contextuality with linear programming tools}\label{appendix: loop_bench}

\begin{figure}[t]
    \centering
    \includegraphics[width = 0.95\textwidth]{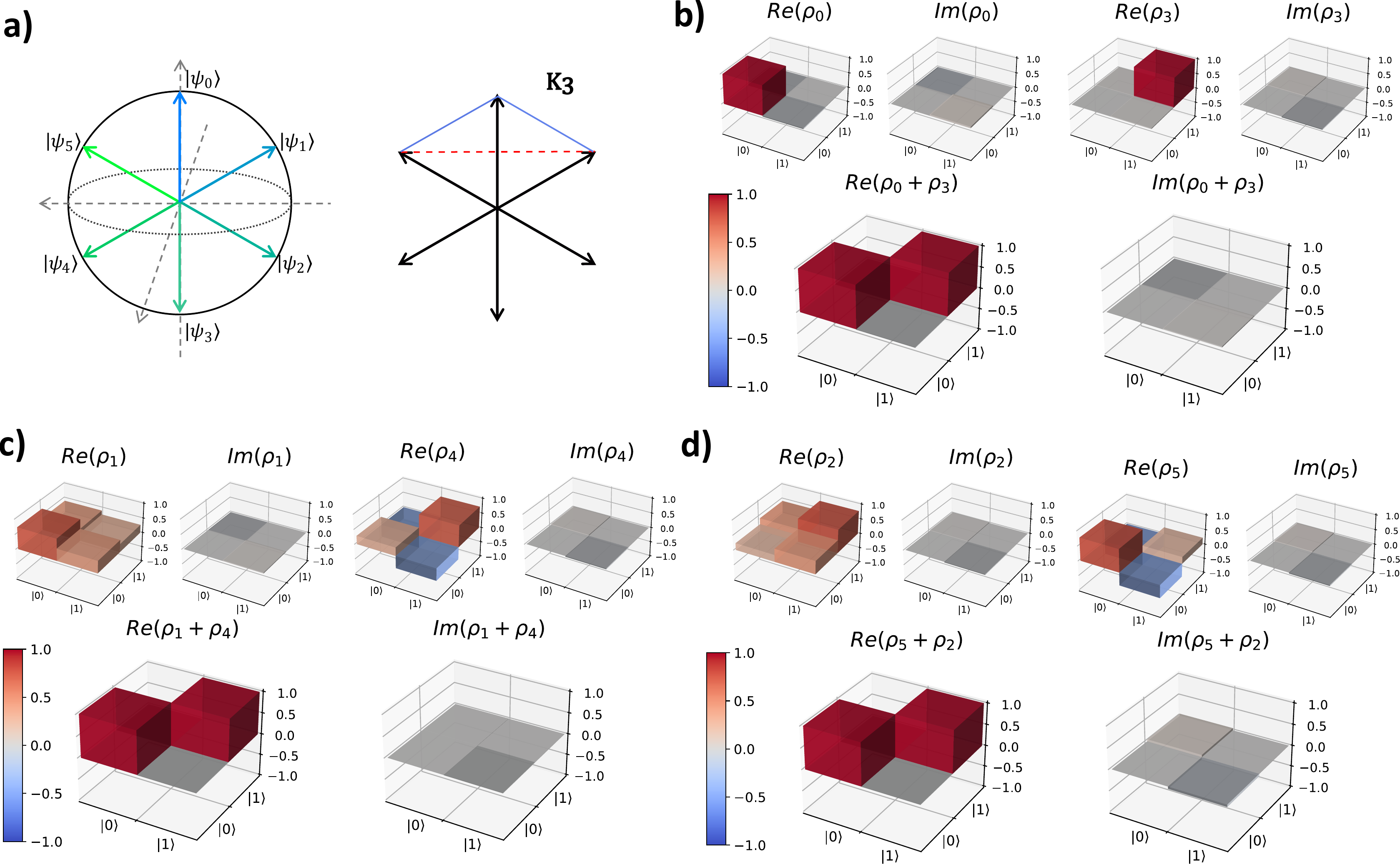}
    \caption{\textbf{Quantum state tomography for states of the noncontextuality inequality.} In (a) we consider the states that are present in the prepare-and-measure scenario, and highlight those used for obtaining a violation of the (robust) $h_3$ inequality from the $K_3$ graph. From (b)-(d) we present the results of the quantum state tomography, as well as confirmation that those states satisfy the operational equivalences Eq.~\eqref{eq: op equivalences} within experimental error.}
    \label{fig:tomo}
\end{figure}

\subsection{Generalized contextuality from the algebraic approach: inequality violation}

We have performed quantum tomography for the states that are generated by both {beam splitters} inside the MZI used in the quantum interrogation test, namely $\mathcal{S}_{ideal} := \{\vert 0\rangle, \vert \theta \rangle, \vert -\theta \rangle,\vert 1 \rangle, \vert \theta^\perp \rangle,\vert -\theta^\perp \rangle\}$ that correspond to density matrices $\mathcal{S} := \{\rho_0,\rho_1,\rho_2,\rho_3,\rho_4,\rho_5\}$, for the value $\theta=5\pi/6$ resulting in states in an hexagon of the rebit subspace of the Bloch sphere. The resulting states are depicted in Fig~\ref{fig:tomo}. As one of the assumptions present in the description of the scenario, and relevant for the derivation of inequality \eqref{eq: noncontextuality inequality}, is the assumption that operationally the states satisfy the operational equivalences $\rho_0 + \rho_3 = \rho_1+\rho_4=\rho_2+\rho_5$. We probe these equivalences by performing quantum state tomography and show that these are correct up to experimental error. In this case, the noncontextuality inequality takes the form
\begin{equation}
    h_3^{(robust)} := \text{Tr}(\rho_1\rho_0)+\text{Tr}(\rho_2\rho_0)-\text{Tr}(\rho_2\rho_1)-\text{Tr}(\rho_0\rho_3)-\text{Tr}(\rho_1\rho_4)-\text{Tr}(\rho_2\rho_5)\leq 1. \label{eq: quantum inequality violation}
\end{equation}

Using the MZI to probe this specific inequality we obtain {$h_3^{(robust)} = 1.178\pm 0.012>1$. The violation of this inequality is a robust witness of quantum generalized contextuality.

Despite the fact that we are able to witness generalized contextuality using the above inequality, the no-go result for the gain in the efficiency $\eta$ compared to noncontextual models is more sensitive to noise. In order to improve this gain, we consider a state-of-the-art approach to witnessing contextuality that uses general probabilistic theories (GPTs), as we now discuss.

\subsection{Generalized contextuality from the geometrical approach: Lack of simplex embeddability}

In order to improve the robustness of depolarizing noise for the specific PM scenario we can consider the GPT/geometric framework. The proof of contextual advantage from Ref.~\cite{wagner2022coherence} relies on a specific characterization of PM scenarios using specific operational equivalences. In our experiment, these equivalences correspond to those present in Eq.~\eqref{eq: op equivalences}, and we have probed their validity within experimental uncertainties. Due to the specification of target equivalences, the robustness to noise in our scenario is \textit{smaller} than what would be possible in case one considers \textit{all} possible operational equivalences satisfied by quantum theory viewed as a GPT, instead of an operational-probabilistic theory. 

We now use the framework studied in Refs.~\cite{selby2021accessible,selby2022opensource,rossi2022contextuality} to specify the accessible GPT fragment of quantum theory and probe its robustness to noise. We will consider that the reader is familiar with the GPT framework, and we refer to the introductory Refs.~\cite{hardy2001quantum,barrett2007information,muller2021probabilistic,plavala2021general}. In this perspective, we consider quantum theory as a GPT, in which case quantum states and measurement effects are considered GPT states and GPT effects.

Remarkably, the geometry approach greatly simplifies the minimal requirements to test contextuality, which corresponds to probing the operational equivalences as it considers, by default, all the possible equivalences present in the \textit{accessible} fragment described by a GPT. This fragment is described simply by: 1) a set of quantum states, 2) a set of measurement effects, 3) a unit effect of the accessible GPT, that in quantum theory simply corresponds to the identity and 4) a maximally mixed state. When one assumes quantum theory, 3 and 4 are assumed to be both the identity and the maximally mixed state $I_d/d$. In order to probe generalized contextuality in such a framework, one must test if the accessible fragment just described allows for being embedded in \textit{some} classical GPT, that is a simplicial GPT defined in a possibly higher dimensional GPT space. It was shown that this problem can be written as a linear program~\cite{Gitton2022solvablecriterion,selby2021accessible,selby2022opensource}. 

For the set of effects, we choose the set of tomographically complete measurements over the qubit space $\mathcal{M} := \{X,Y,Z \}$. The pair $\mathcal{A} \equiv (\mathcal{S},\mathcal{M})$ completely characterizes the accessible GPT we investigate. The linear program from Ref.~\cite{selby2022opensource} returns robustness to depolarizing noise for the fragment $\mathcal{A}$ and it is given by $\nu(\mathcal{A}) = 0.1121$. Two comments are necessary: first, this robustness is considered \textit{with respect to} $\mathcal{A}$, and hence with respect to the noisy states, as opposed to the robustness considered before with respect to the ideal states. If we consider the robustness with respect to the ideal states using $\mathcal{A}_{ideal} = (\mathcal{S}_{ideal},\mathcal{S}^\dagger_{ideal})$ we have that $\nu(\mathcal{A}_{ideal}) = 0.333$, where $\mathcal{S}^\dagger_{ideal}$ are the same elements of $\mathcal{S}_{ideal}$, now viewed as accessible effects. Second, for the quantum interrogation, we assume that the errors in the generation of coherence from the {beam splitter} are significantly higher than the errors in the detection apparatus. Finally, we see that robustness to depolarizing noise is \textit{much larger} when one considers the GPT perspective as many more equivalences are used, instead of simply considering those from Eq.~\ref{eq: op equivalences}, as expected. 

Because we have obtained a value $\nu(\mathcal{A}) = 0.1121$ for the fragment using the states  considered in Fig.~\ref{fig:tomo}, i.e., experimentally probed, this is once more a  witness of generalized contextuality in our device, provided now by the geometric approach.


\subsection{The role of contextuality in the efficiency of $\eta$}

Let us compare the efficiency found experimentally with the smallest possible noisy efficiency $\eta$ for the quantum interrogation that still allows for a claim of quantum advantage.  We consider the specific value of efficiency $\eta$ reached using the states from Fig.~\ref{fig:tomo}. The results are presented in Table~\ref{tab: efficiency contextuality}.

\begin{table}[ht]
    \centering
    \begin{tabular}{c|c|c|c}
        & Robustness to depolarization ($\nu$) & Quantum Efficiency ($\eta$) & Worse case efficiency ($\eta(\nu)$)\\
        \hline
        PM Ideal & 0 & 0.428571 & 0.285714\\
        \hline 
        PM Robust & 0.057 & 0.428571 & 0.419385 \\
        \hline 
        GPT Robust & 0.333 & 0.428571 & 0.379353\\
        \hline 
        Experiment & 0.112 & 0.428$\pm$ 0.006 & 0.410757
    \end{tabular}
    \caption{\textbf{Efficiency of the quantum interrogation task.} Analysis of the efficiency $\eta$ of quantum interrogation with respect to the states tomographically characterized in Fig.~\ref{fig:tomo}. We consider the ideal quantum efficiency obtainable with the ideal set of states/effects for the PM scenario, and its robust counterpart, in the two first rows. Third row corresponds to the analysis for the GPT approach.  Robustness to depolarization is calculated, as well as its effect on the efficiency $\eta$ as described with Eq.~\eqref{eq: robust efficiency}. Efficiencies higher than those in  the third column cannot be explained with noncontextual models.}
    \label{tab: efficiency contextuality}
\end{table}

We note that due to sensitivity of $\eta$ to noise $\nu$, the algebraic approach allows   noncontextual models to reproduce the  efficiency close to the one experimentally found, even though smaller within experimental error. 
Nevertheless, to better observe that the quantum efficiency found cannot be reproduced by noncontextual models for the fragment of quantum theory considered, we can use the GPT/geometrical  approach. The last row in Table~\ref{tab: efficiency contextuality} presents the efficiency $\eta$ found experimentally, the robustness $\nu$ obtained using the linear program from Ref.~\cite{selby2022opensource} for the quantum states characterized by tomography (see Fig.~\ref{fig:tomo}) and the smallest value for which any quantum efficiency would still incur into quantum advantage for that particular noise level. The other rows in Table~\ref{tab: efficiency contextuality} correspond to numerical bounds comparing optimal quantum results, the noise level that would destroy the quantum advantage argument and the corresponding quantum efficiency for that noise.
 
A noncontextual model would be able to reproduce the data, and hence the efficiency of the task, only in the case that: (i) the effective noise model cannot be mapped as a depolarizing channel, neither be considered as a worse-case noise model through robustnesss to depolarization, implying that we cannot use the program from Ref.~\cite{selby2022opensource}, (ii) the noncontextual model is somewhat conspiratorial, using precisely some non-trivial aspects of the noise to pass the test of non-embeddability characterized by robustness to depolarization, while being simplex-embeddable, for the states in Fig.~\ref{fig:tomo} (iii) the noncontextual model is somewhat conspiratorial in a different sense, using the fact that we do not \textit{perfectly match} the operational equivalences, but only respect them within small experimental errors, and uses this fact to pass the tests considered, i.e., inequality violation and linear-program, and reproducing the high efficiency found, (iv) the states found in Fig.~\ref{fig:tomo} and those used to perform the interrogation task are not the same, even though they were probed within the exact same system. Such situations are highly discredited. In summary, Table~\ref{tab: efficiency contextuality} corroborates the hypothesis that we have  experimentally witnessed efficiencies for the quantum interrogation task that could never be reached by noncontextual ontological models. The last row of the table shows the actual experimental results, for the efficiency obtained experimentally and for the robustness to depolarization obtained from quantum state tomography. 

\subsection{Loopholes}

We make a small note on the loopholes associated with our test of contextuality. It is well established that generalized contextuality resolves various issues regarding common loopholes in experimental tests of no-go results. Among experimental aspects that constitute loopholes for testing the standard notion of Kochen-Specker contextuality or Bell's notion of local causality are: 1) sharpness of measurements, 2) statistical independence, 3) photon detection and 4) freedom of choice. \textit{None} of those constitute loopholes for testing generalized contextuality. 

However, some loopholes persist in our experiment, these are:

\begin{enumerate}
    \item[(a)] Assumption of quantum theory. General tests of contextuality should, in principle, be theory-independent. Some might think that it is necessary to probe contextuality without assuming that quantum theory corresponds to the underlying GPT.
    \item[(b)] Robustness analysis beyond depolarizing noise. It is likely that, considering resource theoretic arguments, depolarizing noise corresponds to the most detrimental type of noise to test generalized contextuality. For instance, Ref.~\cite{rossi2022contextuality} showed proofs of contextuality with arbitrarily large dephasing noise. The tools developed in Ref.~\cite{selby2022opensource} do not allow for considering generic experimental imprecision's, i.e., uncertainty in the states and effects. 
    \item[(c)] Tomographic completeness. Tests of generalized contextuality need the assumption that the fragment (or the operational equivalences) has tomographically complete sets of operations. We do not discuss this loophole here. We have assumed a tomographically complete set of measurements.
\end{enumerate}

It is worth mentioning the following. We do not use the methodology of secondary procedures, as described in Refs.~\cite{mazurek2016experimental,schmid2018contextual} for imposing the symmetries of the operational PM scenario. However, the only symmetry used for violating Eq.~\eqref{eq: quantum inequality violation} was that $r_{i,j}=r_{j,i}$. We study this in depth for 4-dimensional states in Appendix~\ref{app: rij=rji}. For qubit violations, this source of error is smaller than what would be necessary to lose the violation. Note that this is \textit{not} a loophole for the analysis using the GPT approach. Therefore, this is not a loophole in our test for contextuality overall, and only of our considerations using the operational PM scenarios.

\section{Coherence witnesses tailored for two-level systems}
\label{appendix:MZI}

\begin{figure}
    \centering
    \includegraphics[width=0.6\textwidth]{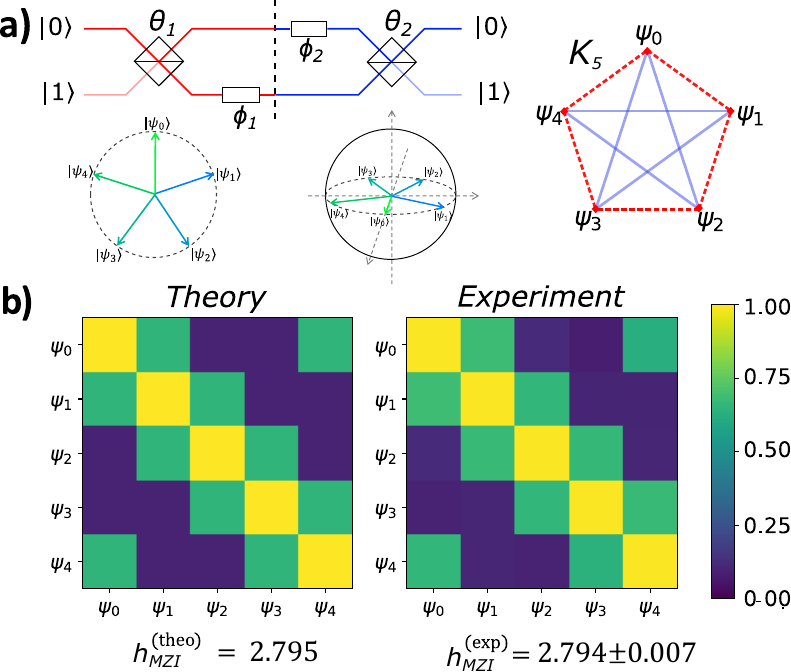}
    \caption{\textbf{Symmetric quantum states witnessing coherence in the Mach-Zehnder from high inequality violation.} a) Scheme for qubit encoding and measurement. The first tunable {beam splitter} $\theta_1$ and phase shift $\phi_1$ prepare the qubit $\ket{\psi_i}$. The second pair $(\theta_2$, $\phi_2)$ performs the projection on $\bra{\psi_j}$. The violation is maximized by qubit states equally spaced on a great circle of the Bloch sphere. The inequality is schematized by the graph in the figure. Blue edges are overlaps that must be summed and in red the ones to be subtracted. b) Overlaps $r_{ij}$ for the maximum violation set on the Bloch sphere equator. }
    \label{fig:expbomb}
\end{figure}
An extensive study of all the possible inequalities bounding the coherence-free sets of overlaps of up to 6 states was done in Ref.~\cite{Wagner2022}. Among these, some inequalities stand out  due to their conceptual relevance and robustness to noise. For instance, as was proposed in Ref.~\cite{wagner2022coherence}, the following inequality witnesses coherence inside Mach-Zehnder interferometers {(MZIs) 
, in a robust and efficient way (see Fig.~\ref{fig:expbomb}a for a description of this set-up)}:
\begin{align}
&h_{\mathrm{MZI}}(r) = r_{0,1}+r_{0,4}+r_{1,2}+r_{2,3}+r_{3,4}\nonumber\\
    &\hspace{3em}-r_{0,2}-r_{0,3}-r_{1,3}-r_{1,4}-r_{2,4}\leq 2. \label{ineq: pentagon violation inequality}
\end{align}
It was shown {that this inequality has} a large quantum violation, $\sim 0.795$, and it applies to the common scenario of an MZI with balanced {beam splitters}, using 10 different choices of internal phase shifters.

Our experimental investigation regards such a coherence witness tailored for two-dimensional systems (qubits), reported in the inequality \eqref{ineq: pentagon violation inequality}. We remind that the test 
requires the estimation of overlaps $r_{i,j}=\vert \langle \psi_j \vert \psi_i \rangle \vert^2$. Hence, we prepare the {programmable UPP} to implement the scheme depicted in Fig.~\ref{fig:expbomb} (a). We consider only two modes 
that individuate the MZI sketched in the figure. In the preparation stage, highlighted in red, we inject a photon in the modes 0 
to prepare qubit state $\ket{\psi_i} = \cos\theta_i \vert 0 \rangle + e^{i\phi_i}\sin\theta_i\vert 1 \rangle$ by programming the two angles $\phi_1=\phi_i$ and $\theta_1=\theta_i$.
Conversely, in the measurement stage, highlighted in blue, the second phase shifter $\phi_2$ is set to $-\phi_j$ and 
$\theta_2 = \theta_j$ to realize the projection onto $\bra{\psi_j}$. The overlap $r_{i,j}$ between the two states is given by the probability to detect a photon in mode 0. 
The inequality functional reported in the l.h.s of  Eq.~\eqref{ineq: pentagon violation inequality} is represented by the graph in Fig.~\ref{fig:expbomb} (a). Red (dashed) edges correspond to overlaps weighted by a $-1$ phase, while light blue (full) edges correspond to overlaps weighted by a $+1$ phase in the inequality. 
The maximum 
violation 
is 
$5\sqrt{5}/4 \approx 2.795 $. This bound is saturated, for example, by 
states distributed on the vertices of a regular pentagon that lies on the equator of the Bloch sphere as shown in Fig.~\ref{fig:expbomb} (a). The bound for coherence-free states is the value on the r.h.s. of the inequality, that is, 2.
We experimentally measure all the possible overlaps $r_{i,k}$  with $i,k \in \{0,\dots,4\}$ of the states in a set with maximum violation, by tuning the parameters $\theta_i = \pi/4$ 
and $\phi_k=k\cdot 2\pi/5$, for all $k$. In Fig.~\ref{fig:expbomb} (b) we report the matrix of the overlaps compared to the theoretical one. The corresponding 
violation is equal to $2.794 \pm 0.007$ and it is consistent with the theoretical one within one standard deviation. We estimated the violations by considering only the upper triangular part of the overlaps matrix, i.e $r_{i,j}$ with $i>j$. These results prove the generation of (basis-independent) coherence inside the MZI.

\section{{UPP}'s settings for preparing and measuring qudits}
\label{appendix:qudit_settings}
In the following, we report how qudit states are generated and projected in the {programmable UPP}. We also show the circuit settings to obtain the maximum violation of the coherence witnesses.

A generic qutrit state can be encoded from a {single photon} in mode 0 by setting the parameters of the circuit $\{\theta_1,\theta_2, \phi_1, \phi_2\}$ in Fig. 4a to generate the following state 
\begin{equation}
    \ket{\psi} =  \cos{\theta_{1}}\cos{\theta_{2}}\ket{0}+\sin{\theta_{1}}\cos{\theta_{2}}e^{i\phi_{1}}\ket{1}
    +\sin{\theta_{1}}\sin{\theta_{2}}e^{i\phi_{2}}\ket{2}
\end{equation}
The maximum violation of $h_4$ has been reached by the qutrits with the following amplitudes on the $\{\ket{0}, \ket{1}, \ket{2}\}$ basis:

 \begin{equation}
 \left\{
   \begin{aligned}
   \ket{\psi_0}&=\{1, 0, 0\}\\
   \ket{\psi_1}&=\Big\{\frac{\sqrt{5}}{3},\frac{2}{3}, 0 \Big\}\\
    \ket{\psi_2}&= \Big\{\frac{\sqrt{5}}{3},\frac{1}{3}, i\frac{1}{3} \Big\}\\
    \ket{\psi_3}&=\Big\{\frac{\sqrt{5}}{3},-\frac{1}{3}, -i\frac{1}{3} \Big\}
    \end{aligned}\right.
 \end{equation}

Fig. 4b of the main text reports the circuit to encode a generic ququart in a {single photon} that enters from mode 1 of the device. The qudit expressed through the parameters $\{\theta_1, \theta_2, \theta_3, \phi_1, \phi_2, \phi_3\}$ is

\begin{equation}
    \ket{\psi} = \cos{\theta_{2}}\cos{\theta_{1}}\ket{0}+\sin{\theta_{2}}\cos{\theta_{1}}e^{i\phi_{1}}\ket{1}+\sin{\theta_{1}}\cos{\theta_{3}}e^{i\phi_{2}}\ket{2}+\sin{\theta_{1}}\sin{\theta_{3}}e^{i\phi_{3}}\ket{3}
\end{equation}
and the amplitudes of the set of states for maximizing the $h_5$ violation expressed in the computational basis are
 \begin{equation}
 \left\{
   \begin{aligned}
   \ket{\psi_0}&=\{0.61, 0.16, 0.41, -0.65\}\\
   \ket{\psi_1}&=\{0.13,0.12, 0.95, -0.26 \}\\
    \ket{\psi_2}&= \{-0.99,0.01, 0.12, -0.01 \}\\
    \ket{\psi_3}&=\{-0.23, 0.43, 0.05, -0.87 \}\\
     \ket{\psi_4}&=\{0.26, 0.76, -0.03, -0.59 \}
    \end{aligned}\right.
 \end{equation}

The circuit employed for the 5-mode qudits is not universal in the sense that it cannot generate  generic states. The limitation was imposed by the number of layers of MZIs in the 6-mode {UPP} that was not enough to individuate two separated and independent preparation and measurement stages. The 5-mode states that we could prepare and measure are parameterized by 
\begin{equation}
    \ket{\psi} =  \sin{\theta_1} \cos{\theta_2} \sin{\theta_4}\ket{0}
     +\sin{\theta_1}\cos{\theta_{2}}\cos{\theta_{4}}\ket{1}+
    \sin{\theta_{1}}\sin{\theta_{2}}e^{i\phi_{1}}\ket{2}
    +\cos{\theta_{1}}\sin{\theta_{3}}e^{i\phi_{2}}\ket{3}+
    \cos{\theta_{1}}\cos{\theta_{3}}e^{i\phi_{3}}\ket{4}
\end{equation}
where $\{\theta_1, \theta_2, \theta_3, \theta_4, \phi_1, \phi_2, \phi_3 \}$ are the angles of the tunable {beam splitters and phase shifters} reported in Fig. 4c and the input mode is $\ket{2}$. The maximum violation of $h_6$ is reached by the following states, described by the parameters:
 \begin{equation}
 \left\{
   \begin{aligned}
   \ket{\psi_0}&=\{0.60, -0.38, 0.43, 0.42, 0.36\}\\
     \ket{\psi_0}&=\{1, 0, 0, 0, 0\}\\
   \ket{\psi_1}&=\{0.20,-0.002, 0.16, 0.13,  0.96 \}\\
    \ket{\psi_2}&= \{0.20 ,-0.16, 0.97, 0, 0 \}\\
    \ket{\psi_3}&=\{0.20, -0.96, 0.006, 0.16, 0.12 \}\\
     \ket{\psi_4}&=\{0.20, -0.006, 0.16, 0.97, 0.003 \}
    \end{aligned}\right.
 \end{equation}

\section{Effects of experimental noise in the inequalities estimation}\label{app: rij=rji}
\label{appendix:exp_noise}
In the main text we reported the measurements of various inequalities aiming at witnessing coherence and the dimension of the Hilbert space. Such quantities require the calculation of overlaps $r_{i,j}=\vert \braket{\psi_j}{\psi_i}\vert^2$ between pairs of states that belong to a given set. The main sources of errors in the experimental calculation are the statistics of photon counts and the imperfections in preparing $\ket{\psi_i}$ and measuring $\bra{\psi_j}$. The uncertainties reported in the main text derive only from the poissonian statistics of single-photon counts. We refer to such errors as $\sigma_c$.  In this appendix we analyse the effects of an imperfect setting of the optical circuits parameters instead. 

The implicit assumption of the $h_n$ inequalities is that the overlap function $r_{i,j}$ is symmetric. In the experiment this assumption could not be perfectly satisfied. For example, the same state $\psi_i$ has to be encoded both as a ket and as a bra. The parts of the circuit dedicated for the preparation and the measurement are separated and involved different variable {beam splitters and phase shifters}. Any small mismatch between the two stages in encoding the same state $\psi_i$ makes the overlaps matrix not symmetric, i.e $r_{i,j}\neq r_{j,i}$. This effect generates errors in the $h_n$ estimation that may lead to over- or under-estimates. The reason is that the effective number of different states in the set is larger than $n$. For example in the experiment estimating $h_5$, considering only the upper triangular part of the overlaps matrix as we did in the main text, the state $\psi_1$ is encoded as a bra for the $r_{0,1}$ calculation but also as a ket in the preparation stage for $r_{1,2}$, $r_{1,3}$ and $r_{1,4}$. An analogous consideration holds for states $\psi_2$ and $\psi_3$. Thus, it follows that the actual number of states involved in the $h_5$ calculation could in principle be up to 8.

In Fig. \ref{fig:errors} we investigate such an effect that follows from errors in the circuit settings that prepare and measure the states for the $h_5$ calculation. We report the experimental estimate $h_5^{exp}= 1.391 \pm 0.011$ in black which is larger than the maximum theoretical value $h_5^{theo}= 1.375$ in red. We consider two types of errors in the parameters $\{\theta_1, \theta_2, \theta_3, \phi_1, \phi_2, \phi_3\}$ that generate the ququarts and in the angles $\{\theta_4, \theta_5, \theta_6, \phi_4, \phi_5, \phi_6\}$ which perform the projection. The first one $\epsilon$ is a relative error in the setting of the angles in the circuit. The second type $\delta$ is an additive error, i.e. a bias in the angles. The red (blue) shaded area indicates the maximum dispersion of $h_5$ at a given error $\epsilon$ ($\delta$). It is evident from such analysis how much $h_5$ is sensitive to small errors in the optical circuits. In orange we highlight the errors value for which the $h_5$ dispersion is equal to 2.5 and 10 $\sigma_c$, the uncertainty associated to the photon counts of the experiment. On one hand, the two plots give an explanation for the value $h_5^{exp}$ that was larger than theoretical bound. On the other, they provide evidence that the individual parameters of the optical circuit were calibrated with relative errors within 0.005 and biases within 0.5°.

\begin{figure}
    \centering
    \includegraphics[width=\textwidth]{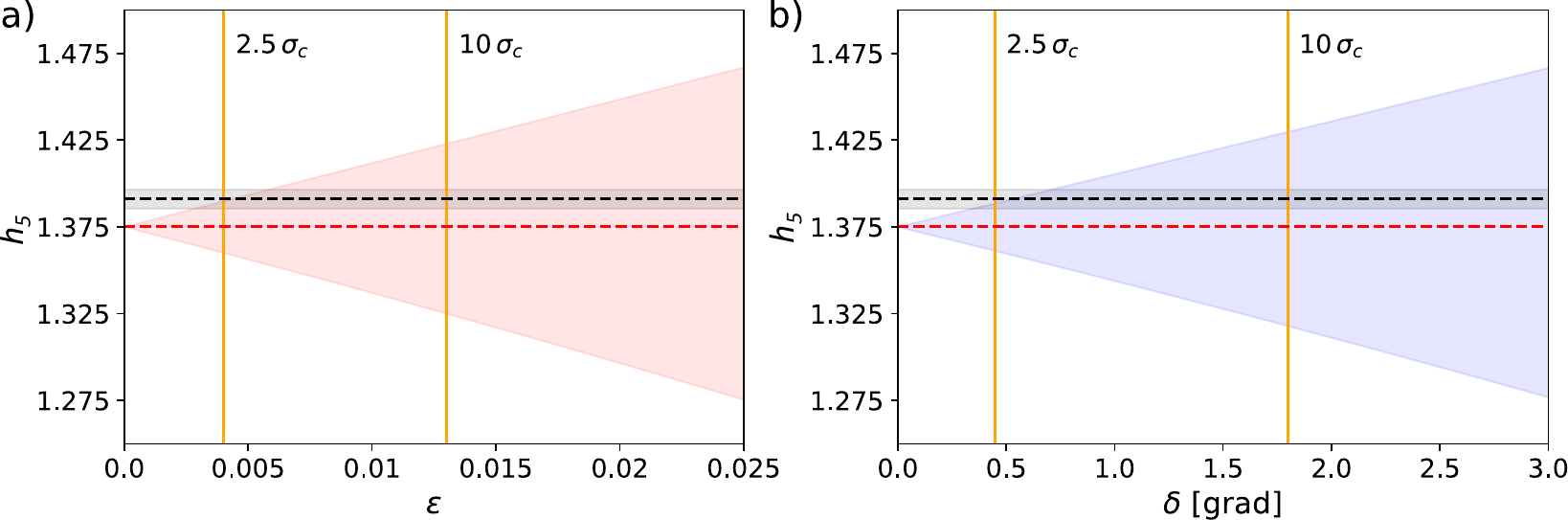}
    \caption{\textbf{Errors in the optical circuit settings.} a) Effect of imperfections summarized in the relative error $\epsilon$ associated to the angles $\{\theta_1, \theta_2, \theta_3, \phi_1, \phi_2, \phi_3\}$ and $\{\theta_4, \theta_5, \theta_6, \phi_4, \phi_5, \phi_6\}$ in the $h_5$ estimation. The red area covers the maximum dispersion associated to $h_5$ for a given level of $\epsilon$. The yellow lines highlight the values of $\epsilon$ which produce errors in the $h_5$ estimation equal to 2.5 and 10 $\sigma_c$, where $\sigma_c$ is the uncertainty due to the single-photon counts statistics. The latter is reported as the grey area around the $h^{exp}_5$ value in black. We report in red the theoretical value $h^{theo}_5$. b) Same analysis for additive errors $\delta$ in the angles of the circuit.}
    \label{fig:errors}
\end{figure}

{
\section{Design, fabrication and calibration of the {UPP}}
\label{app:UPP}

The six-mode {UPP} employed in our experiments consists of a waveguide optical circuit, developed according to the 'universal' design reported in \cite{Clements2016}, which is made {programmable} by means of thermo-optic phase shifters. Details of the circuit layout are shown in Fig.~\ref{fig:6x6WGcircuit}.

The core of the circuit is a rectangular network of {MZIs}, each of them consisting of two directional couplers in cascade. Directional couplers are here designed as two curved waveguide segments, composed only of circular arcs with $R$~=~30~mm curvature radius, which in the middle region are brought close to exchange optical power by evanescent field interaction. The minimum distance is chosen here to produce a splitting ratio of precisely one half. Each  of these {MZIs} is equipped with two thermo-optic phase shifters as shown in Fig. 3c of the Main Text. {The MZI cell is highlighted in Fig.~\ref{fig:6x6WGcircuit} by a red rectangle: its total length is 11.4~mm}.

\begin{figure}
\includegraphics[width=\textwidth]{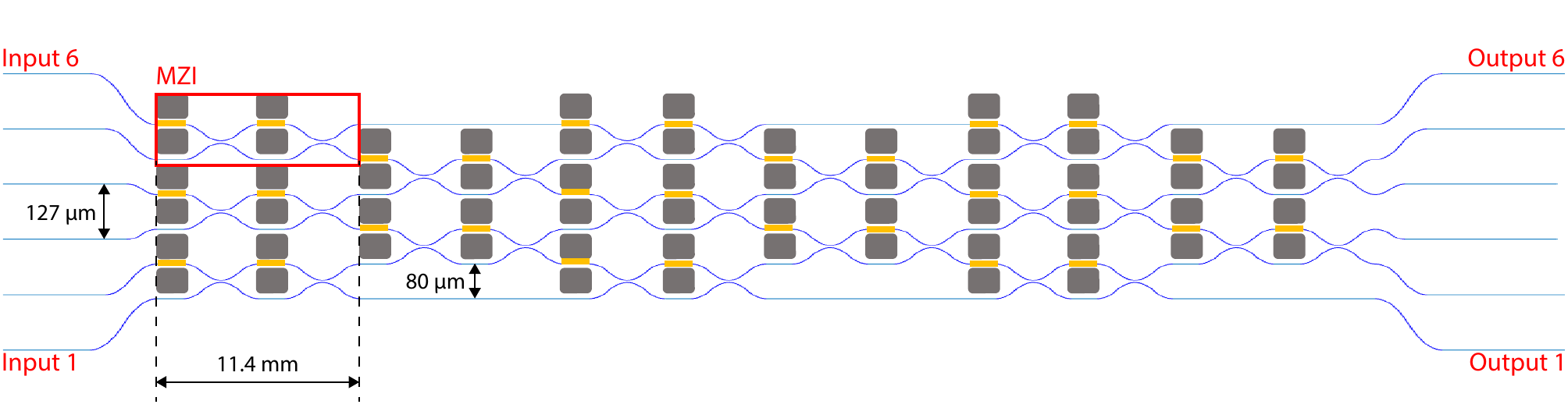}
\caption{\label{fig:6x6WGcircuit} {\textbf{Scheme of the waveguide circuit of the six-mode UPP}. Blue lines depict the waveguide paths, yellow rectangles indicate the microheaters that implement the thermo-optic phase shifters (electrical connection paths are not represented), while dark rectangles indicate the thermal-insulation trenches. The two dimensions are not to scale: the actual circuit length is 82~mm (horizontal dimension in the drawing), but lateral footprint of the circuit is only 0.635~mm (vertical dimension in the drawing).}}
\end{figure}

Lateral pitch of the interferometer arms is {80}~$\mu$m; fan-in and fan-out sections are added at the two ends of the circuits, to match the 127~$\mu$m pitch of standard fiber arrays. The full circuit fits in a $82 \times 20\;\text{mm}^2$ chip.

We fabricated the photonic processor in a commercial {alumino-borosilicate} glass substrate (Eagle XG) with 1.1~mm thickness, using the femtosecond laser writing technology. In particular, we adopted a Yb based laser system (Light Conversion Pharos), delivering a train of pulses at 1030 nm wavelength, repetition rate tunable from 1 kHz to 1 MHz, pulse-width of 170 fs, and average power up to 10.0 W.

Waveguides were directly written 25~$\mu$m below the top surface of the glass substrate, each by 6 overlapped laser scans. For the waveguide inscription the laser repetition rate was set to 1~MHz and 330 nJ pulses were focused by {a 20$\times$ water-immersion objective (0.5 NA)} while the substrate was translated at constant speed of 25.0~mm/s. After irradiation, the substrate underwent a thermal annealing process \cite{arriola2013low, corrielli2018symmetric} which improves propagation losses and optical confinement properties of the waveguides. The resulting optical insertion loss of the whole circuit is about 3~dB.

Thermo-optic phase shifters are based on resistive {microheaters} deposited on the top surface of the optical chip, precisely above the positions of the waveguides where a tunable phase term needs to be introduced. By driving a controlled current into the {microheater}, a controlled temperature increase is achieved locally within the substrate, which in turn induces, by the thermo-optic effect, a proportional phase delay in the light propagating in the waveguide. To fabricate the resistive {microheaters} we used a technique similar to the one described in Ref.~\cite{ceccarelli2019thermal}. Two metallic layers, respectively of chromium (5-nm thick) and of gold (100-nm thick), were deposited in sequence on the substrate surface, by thermal evaporation. Metal deposition was followed by a thermal annealing process in vacuum (300$^\circ$C for 3~h, $p < 1 \cdot 10^{-5}$~mbar) to stabilize the resistivity of the metallic film. Finally, the same femtosecond laser used for waveguide inscription was used to pattern the {microheaters}: here femtosecond laser pulses of 200~nJ energy and 1~MHz repetition rate were focused by a 10$\times$ objective (0.25 NA), while translating the substrate at 2~mm/s constant speed. A single {microheater} has a length of 1.5 mm and a width of 10~$\mu$m, providing an average resistance value of about 125~$\Omega$, and is connected to millimeter-wide contact pads by paths of minor resistance.

To increase the efficiency of the phase shifters, and reduce cross talks between adjacent devices, insulating micro-trenches were also excavated on both sides of each {microheater} \cite{ceccarelli2020low}, using water-assisted laser ablation. In detail, we used the same Yb femtosecond laser system but we set the repetition rate to 20~kHz, and tuned the pulse-compressor to stretch the laser pulses to about 1~ps duration. Ablation was performed by focusing pulses with 1.5~$\mu$J energy with a 20$\times$ water-immersion objective (0.5 NA); scanning speed was 4~mm/s. Each trench has a width of 60~$\mu$m and a length of 1.5~mm, i.e. it is as long as the {microheater}. The glass 'wall' between two trenches, into which the waveguide is inscribed and above which the {microheater} is deposited, is as thin as 20~$\mu$m. For practical reasons, the trench-excavation process followed the waveguide fabrication (namely, it was conducted after waveguide irradiation and annealing) and preceded the deposition of the metallic layers for phase shifter fabrication.
}

{As a final step of the process, the UPP was assembled on an aluminum heat sink along with interposing printed circuit boards and fiber arrays (single-mode at the input and multi-mode at the output) in order to guarantee easy electrical and optical access to the device.

Once fabrication and packaging were complete, the UPP was subject to an extensive calibration procedure with the aim of fully characterizing the relation between electrical currents and phase terms induced throughout the circuit. This calibration process is based on the following model. Under the assumption of perfectly balanced (i.e. 50:50) directional couplers, the normalized optical power $P_{cross}$ at the cross output of the MZI is expressed by
\begin{equation}
	P_{cross} = \frac{1 + \cos(\theta)}{2},
	\label{calib:opt_power}
\end{equation}
where $\theta$ is the internal phase of the MZI, namely the phase term that we tune to change the splitting ratios of the variable beam splitters of the UPP. In principle, due to thermal cross talk effects, the phase delay $\theta_i$ induced on the $i$-th MZI depends on the electrical current $I_j$ flowing through the $j$-th microheater and this dependence is governed by Joule's law of heat dissipation. In formulas:
\begin{equation}
	\theta_{i} = \theta_{0,i} + \sum_{j}\alpha_{ij}I_j^2(1+\beta_{j}I_j^2),
	\label{calib:current}
\end{equation}
where $\theta_{0,i}$ are static phase contributions related to fabrication tolerances, $\alpha_{i,j}$ are the thermo-optic coupling coefficients between microheaters and MZIs and $\beta_{j}$ are correction factors needed to take into account the dependence of the microheaters' electrical resistance on the temperature. Here, the superposition theorem is successfully employed in presence of the nonlinear correction factors $\beta_{j}$ thanks to the reduced cross talks and to the limited dependence of gold electrical resistivity on the temperature. Moreover, it is worth noting that, due to the peculiar geometry of the circuit (see again Fig.~\ref{fig:6x6WGcircuit}) horizontally neighboring microheaters are much further apart (some mm) than vertically neighboring ones (a few hundreds of µm). As verified experimentally, this means that we can neglect the thermo-optic coupling between microheaters and MZIs that are not on the same column, with a great advantage in terms of calibration effort and accuracy.

Each MZI was individually characterized by employing coherent light at 785~nm (Thorlabs L785P25), an external photodetector (Thorlabs PM16-121) and a multichannel phase shifter driver (Qontrol Q8iv). A custom Python software was developed to enable fully automated measurements of the optical power transmission of the MZIs as a function of the electrical current flowing though a given microheater. The isolation procedure reported in \cite{alexiev21calibration} was adopted in order to guarantee that light was always entering only one of the MZI cell’s inputs and that light coming from only one of its outputs was being detected. Optical fiber switches (Lfiber) were employed in order to route the light to a given input of the UPP and to allow the photodetector to collect the light coming from a given output. Finally, Eq.~\ref{calib:opt_power} and \ref{calib:current} were exploited to best fit the experimental data and extract the values corresponding to $\theta_{0,i}$, $\alpha_{i,j}$ and $\beta_{j}$.

Similar models and measurement procedures were employed also for the external phase shifters (i.e. the ones inducing the phase terms $\phi_i$ at the input ports of the variable beam splitters). This was done with the same procedure reported for the internal phases $\theta$, but interferometric rings were in this case purposely formed inside the UPP by exploiting multiple MZIs set to operate as mirrors ($\theta = \pi$) or balanced beam splitters ($\theta = \pi/2$) as reported in \cite{carolan2015universal, harris2017quantum}. The assumption of negligible heat diffusion along the horizontal axis was employed again, but this time phase terms induced vertically by either internal or external phase shifters in waveguide regions that do not correspond to physical phase shifters were remapped onto external phase shifters as horizontal cross talk terms following the algorithm reported in \cite{crespi2022indistinguishability}. 

Once relation \ref{calib:current} is calibrated for all $\theta_i$ and $\phi_i$, it is possible to invert it and thus obtain the set of the electrical currents $I_j$ corresponding to given phase settings $(\theta_i, \phi_i)$. In this work, the full calibration dataset was employed whenever the control of all the programmable MZIs composing the UPP mesh was necessary. In particular, we employed such a global calibration for the 5-mode qudits in the $h_6$ inequality. In all the other cases we preferred to perform ad hoc calibration procedures aimed at controlling a limited part of the processor with minimum experimental errors and, thus, maximum accuracy. 

\begin{figure}
	\includegraphics[width=\textwidth]{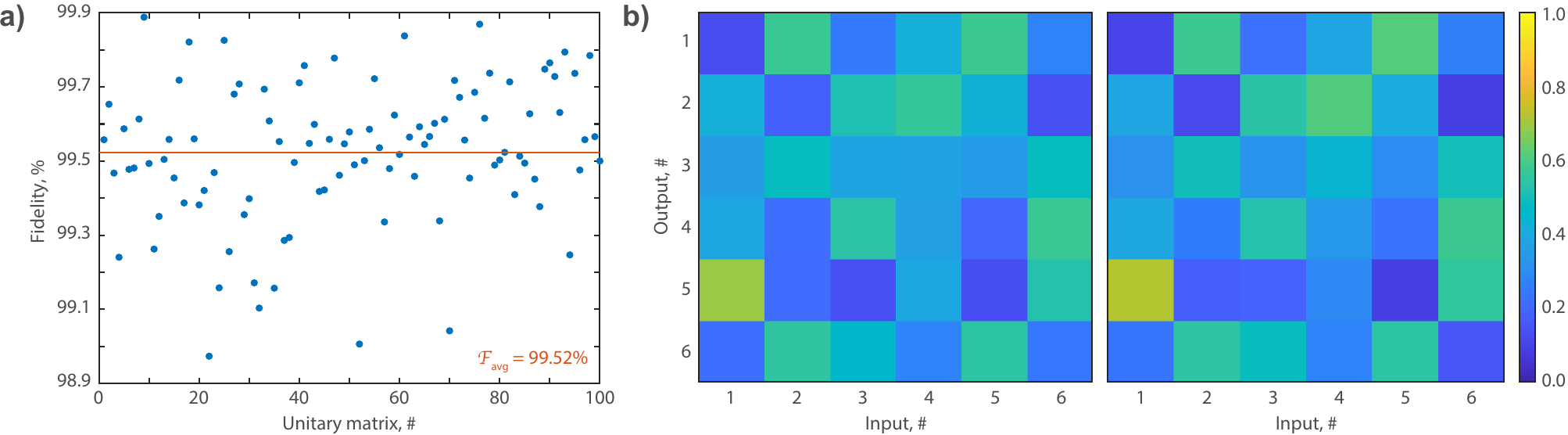}
	\caption{\label{fig:6x6fidelity} {\textbf{Experimental benchmarking of the calibration dataset.} a) Distribution of the fidelities calculated on the moduli of the elements of 100 random unitary matrices. The average fidelity $\mathcal{F}_{avg} = 99.52\%$ is indicated by the horizontal red line. b) Example of unitary matrix (target $|T|$ and experimental $|T_{exp}|$) implemented with accuracy $\mathcal{F}(T, T_{exp}) = 99.55\%$ (the figure shows only the moduli of the matrix elements).}}
\end{figure}
Before starting the actual experiments, the UPP calibration was benchmarked through the following procedure: (I) extraction of a random matrix $T\in U(6)$, (II) calculation of the corresponding phase settings through the decomposition algorithm described in \cite{Clements2016}, (III) implementation of the phase settings in the UPP by exploiting the calibration dataset, (IV) intensity measurements to reconstruct the moduli of all the elements of the experimental matrix $T_{exp}$ \cite{hoch2023characterization} and (V) evaluation of the implementation accuracy through the following definition of fidelity:
\begin{equation}
\mathcal{F}(T, T_{exp}) = \frac{1}{6} \text{Tr}(|T^\dagger| |T_{exp}|).
	\label{calib:fidelity}
\end{equation}
The benchmarking procedure was repeated for 100 different unitary matrices and the final result is reported in Fig. \ref{fig:6x6fidelity} (a). The fidelity was almost entirely distributed over 99\%, with an average value $\mathcal{F}_{avg} = 99.52\%$. An example of unitary matrix implementation with fidelity $\mathcal{F}(T, T_{exp}) = 99.55\%$ is reported in Fig. \ref{fig:6x6fidelity} (b), in which we compare target and experimental matrices $|T|$ and $|T_{exp}|$ (the figure shows only the moduli of the matrix elements).}

\bibliography{bibliography}